\documentclass[manuscript,acmsmall,screen,nonacm]{acmart}

\usepackage{algorithm}
\usepackage[noend]{algpseudocode}
\usepackage{array}
\usepackage{balance}
\usepackage{booktabs}
\usepackage{caption}
\usepackage{multirow}
\usepackage{xcolor}
\usepackage{wrapfig}
\usepackage{xr}

\algnewcommand{\Input}[1]{\Statex \textbf{Input} #1}
\algnewcommand{\Returns}[1]{\Statex \textbf{Returns} #1}

\newtheorem{corollary}{Corollary}
\newtheorem{lemma}{Lemma}
\newtheorem{theorem}{Theorem}

\newcolumntype{H}{>{\setbox0=\hbox\bgroup}c<{\egroup}@{}}

\newcommand{\NAME}{GAS}

\newcommand{\algsize}{\small}

\renewcommand{\paragraph}[1]{\vspace{0.05in}\noindent\textbf{#1}}

\begin{document}

\title{\NAME{}: Generating Fast and Accurate Surrogate Models for Autonomous Vehicle Systems}

\author{Keyur Joshi}
\orcid{0000-0002-5794-6257}
\email{kpjoshi2@illinois.edu}
\author{Chiao Hsieh}
\orcid{0000-0001-8339-9915}
\email{chsieh16@illinois.edu}
\author{Sayan Mitra}
\orcid{0000-0002-6672-8470}
\email{mitras@illinois.edu}
\author{Sasa Misailovic}
\orcid{0000-0001-7319-8845}
\email{misailo@illinois.edu}
\affiliation{
  \institution{University of Illinois Urbana-Champaign}
  \city{Urbana}
  \country{USA}
}

\begin{abstract}

Modern autonomous vehicle systems use complex perception and control components.
These components can rapidly change during development of such systems,
requiring constant re-testing. Unfortunately, high-fidelity simulations of these
complex systems for evaluating vehicle safety are costly. The complexity also
hinders the creation of less computationally intensive surrogate models.

We present \NAME{}, the first approach for creating surrogate models of complete
(perception, control, and dynamics) autonomous vehicle systems containing
complex perception and/or control components. \NAME{}'s two-stage approach first
replaces complex perception components with a perception model. Then, \NAME{}
constructs a polynomial surrogate model of the complete vehicle system using
Generalized Polynomial Chaos (GPC). We demonstrate the use of these surrogate
models in two applications. First, we estimate the probability that the vehicle
will enter an unsafe state over time. Second, we perform global sensitivity
analysis of the vehicle system with respect to its state in a previous time
step. \NAME{}'s approach also allows for reuse of the perception model when
vehicle control and dynamics characteristics are altered during vehicle
development, saving significant time.

We consider five scenarios concerning crop management vehicles that must not
crash into adjacent crops, self driving cars that must stay within their lane,
and unmanned aircraft that must avoid collision. Each of the systems in these
scenarios contain a complex perception or control component. Using \NAME{}, we
generate surrogate models for these systems, and evaluate the generated models
in the applications described above. \NAME{}'s surrogate models provide an
average speedup of $3.7\times$ for safe state probability estimation (minimum
$2.1\times$) and $1.4\times$ for sensitivity analysis (minimum $1.3\times$),
while still maintaining high accuracy.

\end{abstract}

\maketitle

\section{Introduction}
\label{sec:intro}

Autonomous vehicles, such as self driving cars, unmanned aircraft, and utility
vehicles, are increasingly common. They navigate by perceiving the vehicle's
state (position, heading, etc.), making a control decision based on this
\emph{perceived} state, and moving accordingly. To preserve life and property,
developers specify safety properties that the vehicle must satisfy in certain
scenarios. However, 1)~many sensors have a nondeterministic output (e.g., GPS
and LIDAR) and 2)~the vehicle's software processes sensor values or makes
decisions using complex, possibly imperfect components such as neural networks
and lookup tables. Given this uncertainty in the output of the perception and
control systems, proving that an autonomous vehicle will \emph{never} violate
safety properties in a scenario is impossible or impractical. Developers instead
focus on proving that the vehicles will satisfy the safety property \emph{with
high probability}, using probabilistic and statistical techniques.

\emph{Monte Carlo Simulation} (MCS) is perhaps the most used general method for
estimating the probability of entering an unsafe state over time. However, using
MCS requires a large number of resource-intensive system simulations in order to
get a sufficiently accurate estimate of the vehicle state distribution over
time~\cite{menghi2020approximation,robert2004monte}. These costs are
prohibitively high, especially in development tasks that require iterative
modifications or continuous re-testing of the software, including design
exploration, parameter tuning, or system-software regression testing.

\emph{Surrogate Models} aim to provide an accurate and faster replacement for
the original costly model of many engineering systems. These models can be
created using techniques such as abstraction
refinement~\cite{menghi2020approximation}, machine learning~\cite{mlsurrogate},
or Generalized Polynomial Chaos (GPC)~\cite{XiuGPCBook}. For instance, previous
research has created GPC surrogate models for the equations of vehicle
dynamics~\cite{GPCforDynamics}. However, the presence of complex perception and
control components (which often dominate the original model's execution time)
hinders the application of existing surrogate model construction techniques to
creating surrogate models of the complete vehicle system. For example, the
output of a neural network which processes a camera image to perceive the
vehicle's state is not only affected by the ground truth state, but also
multiple environmental parameters (e.g., weather, lighting, nearby objects,
etc.) which may affect the neural network's output in a nondeterministic manner.
Existing techniques struggle to capture this complicated relationship between
environmental parameters and the neural network output.

\paragraph{Our work.} We present \NAME{} (\textbf{G}PC for \textbf{A}utonomous
Vehicle \textbf{S}ystems), the first approach for creating surrogate models of
complete autonomous vehicle systems which compose complex perception and/or
control components with vehicle dynamics. The resulting surrogate models are
close approximations of the original vehicle models. They provide a faster
alternative to MCS when developers experiment with system components, or tune
various system parameters during the design and testing stages of vehicle
development.

\NAME{} first creates a \emph{perception model} to calculate the
\emph{distribution} of error in the output of the perception system for any
ground truth state. \NAME{} directly samples this error distribution, minimizing
the need for costly experimentation with environmental parameters, image
generation, and neural networks. Second, \NAME{} constructs a polynomial
surrogate model of the \emph{complete} vehicle system (perception, control, and
dynamics) using GPC. \NAME{} also supports systems that contain categorical
variables (e.g., in the control system) using an ancillary model, overcoming a
limitation of GPC. Lastly, because \NAME{}'s perception model is created
independently of downstream components, it can be reused when designers alter
vehicle control and dynamics properties during development of the vehicle,
potentially saving a significant amount of time.

We demonstrate the advantages of using the \NAME{}-generated surrogate model in
two applications. First, we use the surrogate model to \emph{estimate the
probability that the vehicle will reach an unsafe state over time} in five
realistic scenarios. These scenarios model systems used in crop management
vehicles, self driving cars, and unmanned aircraft, with associated safety
properties. Each system uses a complex perception component (ResNet-18 or
LaneNet) or a complex control component (neural network controllers or lookup
tables). We show that the probability of reaching an unsafe state calculated by
the surrogate model closely matches that calculated via the original model for
97\% or more of time steps, while being $3.7\times$ faster on average (minimum
$2.1\times$).

Second, we use the surrogate model for \emph{global sensitivity analysis of the
vehicle system} to initial state perturbations for the same scenarios by
calculating Sobol sensitivity indices~\cite{sobolSens} $1.4\times$ faster on
average (minimum $1.3\times$) with an average error of 0.0004 (maximum 0.06).
Lastly, we also investigate how various hyperparameters of the perception model
affect \NAME{}'s overall accuracy and speedup.

\paragraph{Contributions.} This paper makes several main contributions:
\begin{itemize}

\item \textbf{Surrogate models for complex autonomous vehicle systems.} We
present \NAME{}, a novel approach for using GPC to create fast and accurate
surrogate models of complete autonomous vehicle systems with complex perception
and control components.

\item \textbf{Perception models for GPC.} We present models for complex
perception systems that estimate the error distribution of the original systems,
and use these perception models as part of our \NAME{} approach.

\item \textbf{Implementation.} We implement \NAME{} as a tool which automates
the creation of the perception model and the creation and usage of the GPC
surrogate model. We plan to release \NAME{} as an open-source tool in the
future.

\item \textbf{Evaluation.} We evaluate \NAME{} on five realistic scenarios that
model self driving cars, unmanned aircraft, and crop monitoring vehicles.
\NAME{} provides fast and accurate estimates of the probability that the vehicle
will reach an unsafe state and the global sensitivity indices of the vehicle
system.

\end{itemize}

\section{Example}
\label{sec:ex}

\begin{figure}
  \centering
  \begin{minipage}{0.3\textwidth}
    \centering
    \includegraphics[width=\textwidth]{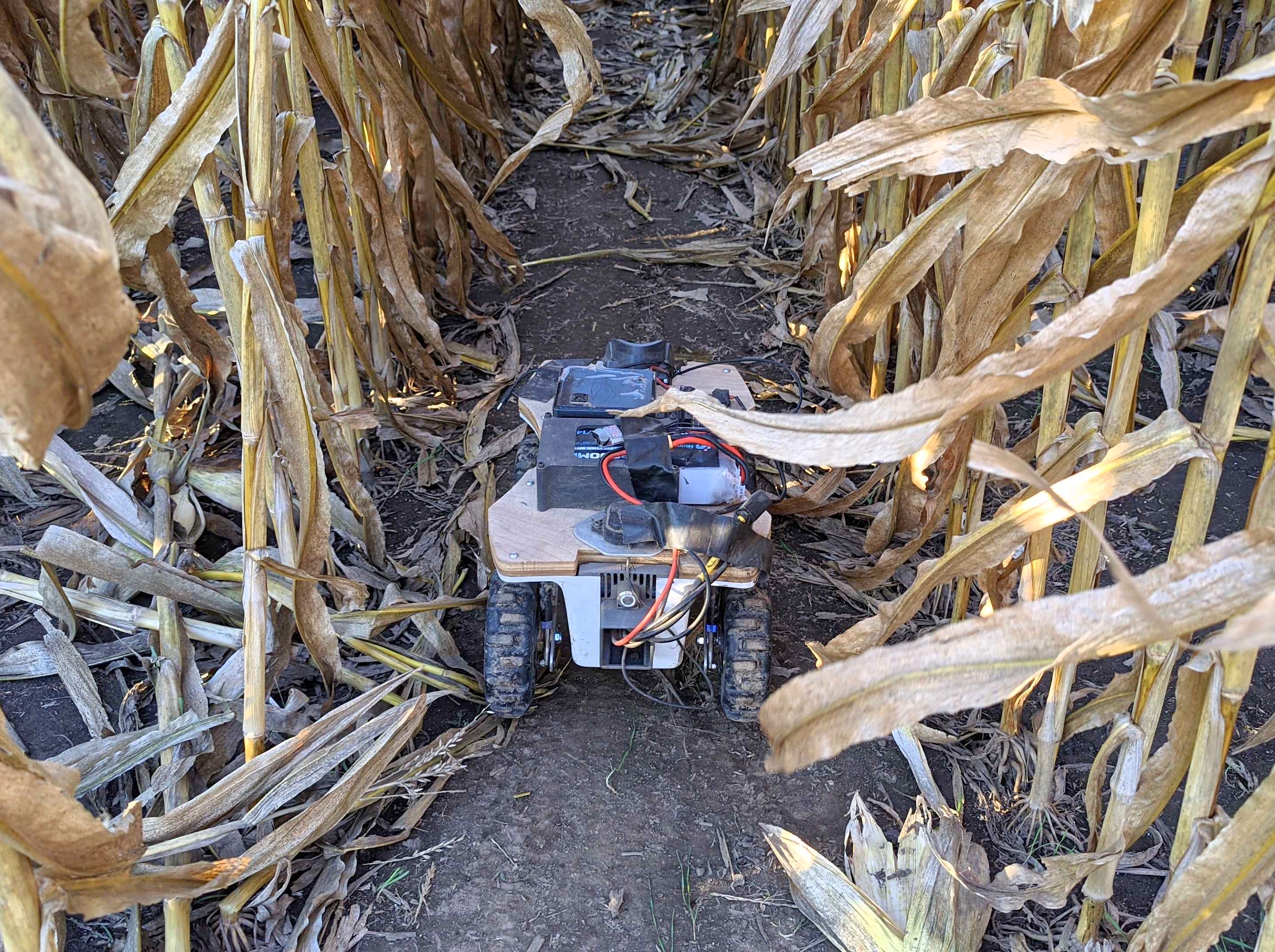}
    \caption{Real-life crop monitoring vehicle}
    \label{fig:ex:irlBot}
  \end{minipage}\hfill
  \begin{minipage}{0.3\textwidth}
    \centering
    \includegraphics[width=\textwidth]{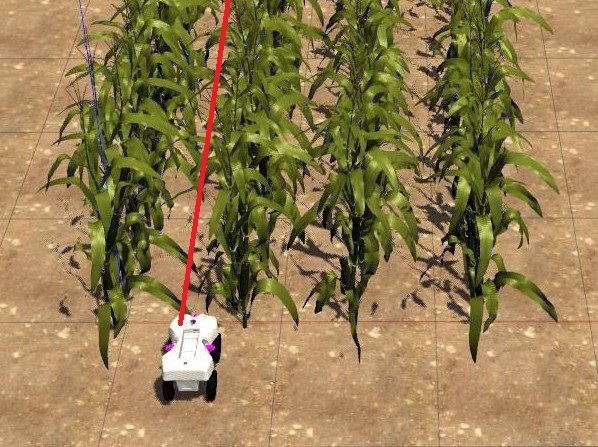}
    \caption{Gazebo simulation -- overhead view}
    \label{fig:ex:overCam}
  \end{minipage}\hfill
  \begin{minipage}{0.3\textwidth}
    \centering
    \includegraphics[width=\textwidth]{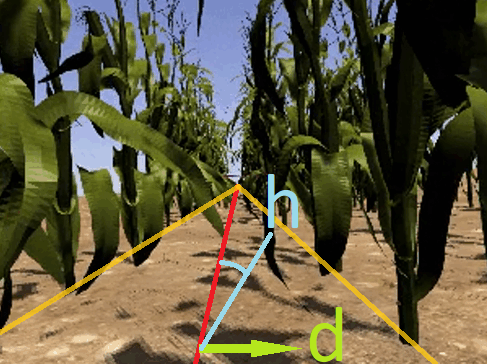}
    \caption{Gazebo simulation -- front camera view}
    \label{fig:ex:fwdCam}
  \end{minipage}
\end{figure}

Consider an autonomous vehicle which travels between two rows of crops in order
to monitor crop growth and detect weeds. Farmers use this information to adjust
fertilizer and herbicide levels for each location. We have adapted this scenario
from~\cite{sivakumar2021learned}.
Figures~\ref{fig:ex:irlBot}-\ref{fig:ex:fwdCam} illustrate the scenario. The
desired path is shown in red.

The top half of Figure~\ref{fig:ex:blockDiag} shows a block diagram
representation of the system model $M_V$ responsible for driving the vehicle
between two rows of crops. First, a camera captures the area in front of the
vehicle. The image depends on the current vehicle state as well as environmental
conditions. Environmental variables include crop types, crop growth stage, crop
model, and lighting conditions. Together, these variables comprise the
environment state space $\mathbb{D}_E$. The neural network analyzes the image to
perceive the current vehicle state. It is a regression neural network, producing
a single state prediction. The relevant state variables in this scenario are:
1)~The \emph{heading} angle $h$, which is the angle between the vehicle's
current heading and the imaginary centerline between the two rows of crops; $h$
can take values in $[-\pi,\pi]$ radians, with $0$ corresponding to the direction
of the centerline. 2)~The \emph{distance} $d$ of the vehicle from the
centerline; $d$ can take values in $[-0.38,0.38]$ meters, with $0$ corresponding
to the centerline. The vehicle state space is therefore $\mathbb{D}_S =
[-\pi,\pi]_h \times [-0.38,0.38]_d$.

As neural networks are inherently approximate, the state perceived by the neural
network may not be the same as the ground truth. The vehicle uses this
approximate heading and distance reading to calculate a steering angle in order
to keep the vehicle on the centerline. Finally, the vehicle moves according to
its constant speed and commanded steering angle.

\paragraph{Unsafe states.} We wish to avoid two undesirable outcomes: 1)~if $|d|
> 0.228m$, the vehicle will hit the crop stems, and 2)~if $|h| > \pi/6$, the
neural network output becomes highly inaccurate and recovery may be impossible.
To avoid these outcomes, we want the vehicle to remain in the \emph{safe
region}, defined as all states within $\mathbb{D}_S^\textit{safe} =
[-\pi/6,\pi/6]_h \times [-0.228,0.228]_d$. Because the vehicle makes steering
decisions based on approximate data, we cannot be certain that the vehicle will
remain safe. Instead, we answer the question: \emph{What is the probability that
the vehicle will remain safe over a period of time?}

\begin{figure}
  \centering
  \includegraphics[width=0.71\textwidth]{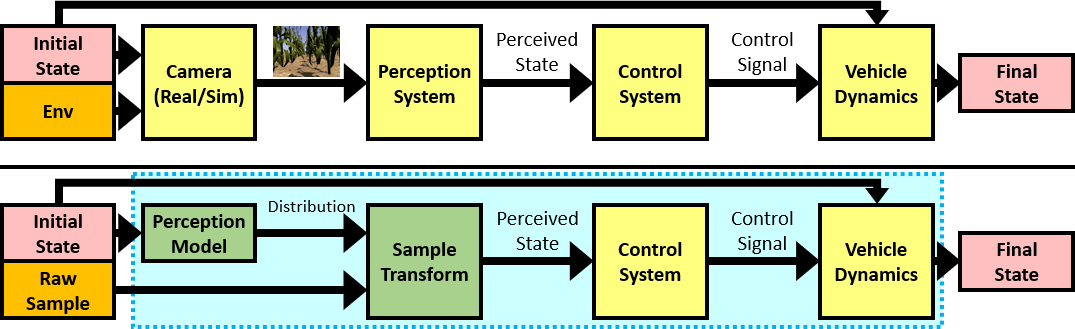}
  \caption{Crop monitor vehicle model: original (top) and abstract (bottom).
            \NAME{} \emph{replaces} the outlined section with a surrogate
            model.}
  \label{fig:ex:blockDiag}
\end{figure}

\paragraph{Monte Carlo Simulation.} In \emph{Monte Carlo Simulation} (MCS), we
simulate the vehicle's movement a large number of times and count the number of
times the vehicle reaches an unsafe state. We use Gazebo
(Figures~\ref{fig:ex:overCam} and~\ref{fig:ex:fwdCam}) for precise control over
the simulation environment.

We simulate 1,000 samples over 100 time steps of 0.1 seconds each. We randomly
sample environmental conditions from an environment distribution $\mathcal{D}_E$
that contains two types of crops, four crop growth stages, crop model
variations, and a range of lighting conditions. We also choose the initial
vehicle state from a normal distribution. For each sample at each time step, we
invoke the autonomous vehicle model, which 1)~captures an image from an
expensive Gazebo simulation in the current vehicle state and environment,
2)~passes the image through the neural network to get the approximate perceived
state, 3)~calculates a steering angle based on the \emph{perceived} state, and
4)~calculates the position of the vehicle after one time step using the
\emph{actual} state and steering angle. At each time step, we count how many
samples are still in a safe state.

\subsection{\NAME{}: Using Generalized Polynomial Chaos}
\label{sec:ex:gpc}

We present \NAME{}, a novel approach for creating surrogate models of complex
vehicle systems using Generalized Polynomial Chaos (GPC). While previous
research (e.g.,~\cite{GPCforDynamics}) has explored using GPC to create models
of vehicle \emph{dynamics}, this is insufficient in our scenario as the process
of capturing and processing the image contributes to over 99\% of the simulation
time. \NAME{} aims to instead use GPC to create a surrogate for the
\emph{entire} vehicle model -- perception, control, and dynamics.

\paragraph{Perception model construction.} The output of the vehicle's
perception neural network depends on the image captured by the front camera.
This image depends not only on the vehicle's state, but also environmental
parameters (e.g., lighting, crop type, crop age, etc.). Given a distribution of
environmental parameters, there is a corresponding output distribution of the
perception neural network for each ground truth state.
Figure~\ref{fig:ex:agbotrealerr} shows the error distribution of the ResNet-18
networks used by this vehicle on the original validation dataset of
\emph{real-world} images collected by the vehicle's camera in a cornfield. The
histogram shows the actual error frequency, while the line shows the fitted
normal distribution. The distribution closely matches the histogram, indicating
that the error of the neural network is normally distributed.

\begin{figure}
\includegraphics[width=.5\columnwidth]{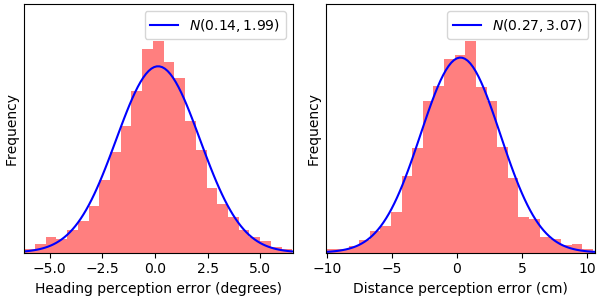}
\caption{Error distribution of the neural network for the validation set of
         \emph{real-world} images.}
\label{fig:ex:agbotrealerr}
\end{figure}

\NAME{} must sample this output distribution when creating the GPC model.
However, 1)~the environmental parameters have a smaller effect on the perceived
state as compared to the vehicle's actual state, and 2)~directly using the
numerous environmental parameters increases the input space over which \NAME{}
must construct the GPC model, which increases model construction time. \NAME{}
therefore abstracts away the actual environmental parameters by creating a
\emph{perception model}. \NAME{} does not require the perception model to be a
perfect abstraction of the neural network\footnote{creating such an abstraction
is at least as hard as solving the neural network verification problem}.
Instead, \NAME{} trains the perception model by selecting an $11 \times 11$ grid
$G$ of ground truth states in the safe region. At each grid point $(h,d)$, it
randomly samples images from $\mathcal{D}_E$, and records the neural network
outputs. It calculates the mean $\mu(h,d)$ and covariance $\sigma^2(h,d)$ of the
outputs at each state. \NAME{} trains a degree 4 polynomial regression model
$M_\mathit{per}$ to predict each component of $\mu$ and $\sigma^2$ at any safe
state. Figure~\ref{fig:ex:agbotpermod} visually compares the neural network
output distribution to the distribution predicted by $M_\mathit{per}$ for images
captured within Gazebo. The red dashed ellipse and the blue solid ellipse show
the $3\sigma$ confidence boundaries for the neural network output distribution
and the distribution predicted by the perception model, respectively. The two
distributions closely match each other, especially when the vehicle is near the
center and pointing straight ahead.

\begin{figure}
\includegraphics[width=.7\columnwidth]{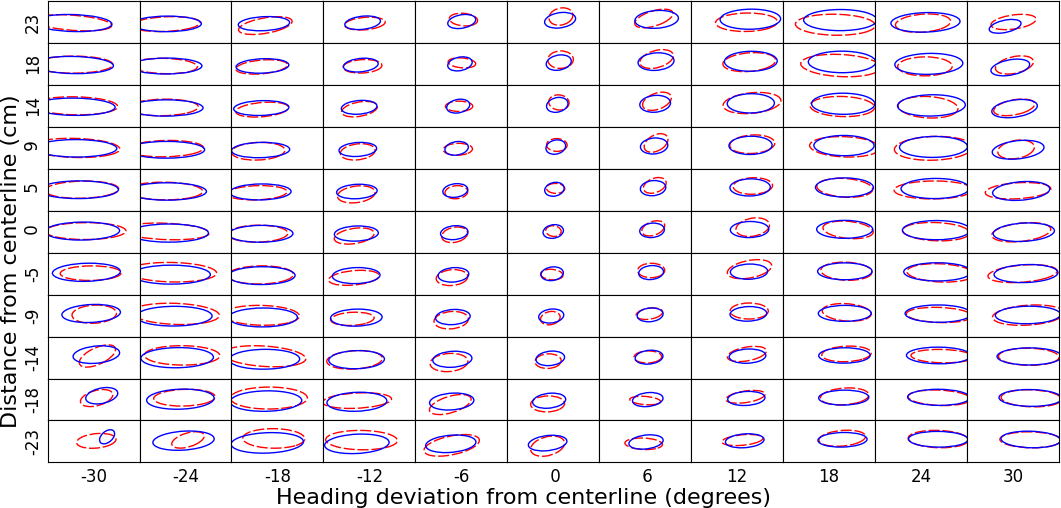}
\caption{Comparison of neural network output distribution (dashed red ellipse)
         to distribution predicted by perception model (solid blue ellipse).
         Each box represents a distinct ground truth state used to construct the
         perception model. The X~and~Y-Axes vary the ground truth heading and
         distance, respectively.}
\label{fig:ex:agbotpermod}
\end{figure}

\paragraph{GPC Surrogate model construction.} Next, \NAME{} creates an abstract
vehicle model $M'_V$, shown in the bottom half of Figure~\ref{fig:ex:blockDiag}.
$M'_V$ first uses the perception model to obtain the neural network output
distribution in the current state. Second, it transforms a sample from a 2D
standard normal distribution into a sample from this output distribution by
multiplying by $\sigma$ and adding $\mu$. The rest of the vehicle model uses
this transformed sample as the perceived state. \NAME{} now uses GPC to create a
polynomial model of the complete system (outlined section of
Figure~\ref{fig:ex:blockDiag}). This model $M_\mathit{GPC}$ is a $4^{th}$ degree
polynomial over 4 variables (2 state variables and a 2D normal distribution
sample) with 70 terms, including 53 interaction terms. \NAME{} \emph{replaces}
the original model in the MCS procedure outlined above with this surrogate model
to estimate the state distribution of the vehicle over time. \NAME{} also
increases the number of samples for distribution estimation to 10,000 for
$M_\mathit{GPC}$ in order to decrease sampling error.

\subsection{Results}
\label{sec:ex:results}

\begin{figure}
  \centering
  \begin{minipage}[t]{0.44\textwidth}
    \centering
    \includegraphics[width=0.49\textwidth]{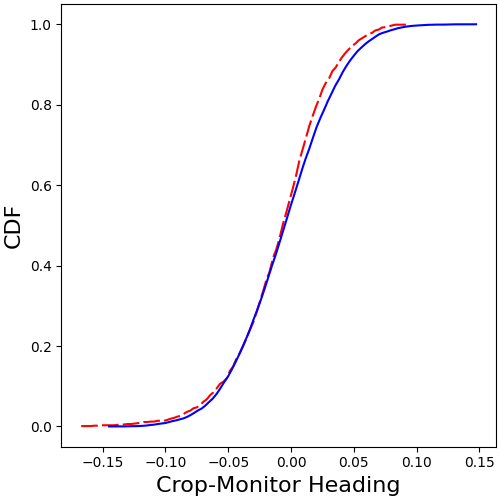}
    \includegraphics[width=0.49\textwidth]{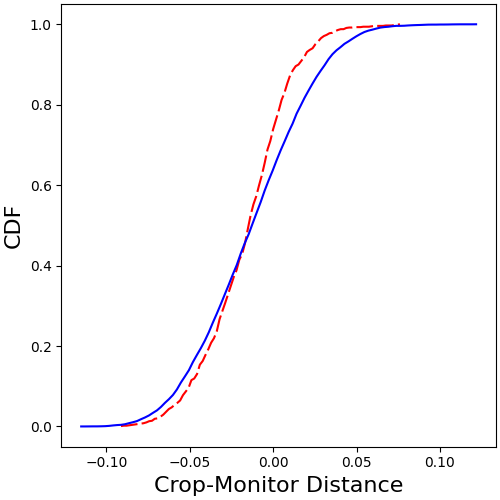}
    \caption{State distribution at the final time step: heading (left) and
             distance from the centerline (right)}
    \label{fig:ex:cdf}
  \end{minipage}\hfill
  \begin{minipage}[t]{0.26\textwidth}
    \centering
    \includegraphics[width=0.85\textwidth]{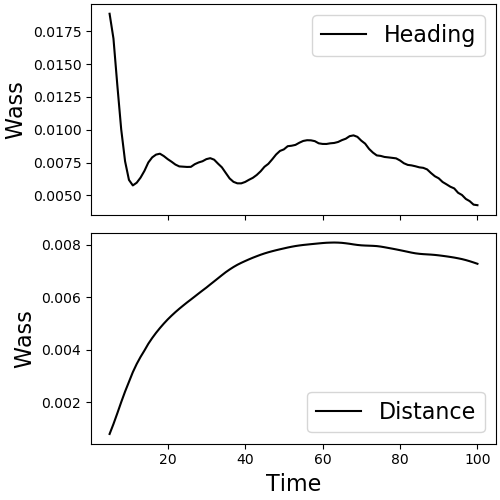}
    \caption{Wasserstein metric (distribution similarity) over time steps}
    \label{fig:ex:wass}
  \end{minipage}\hfill
  \begin{minipage}[t]{0.22\textwidth}
    \centering
    \includegraphics[width=\textwidth]{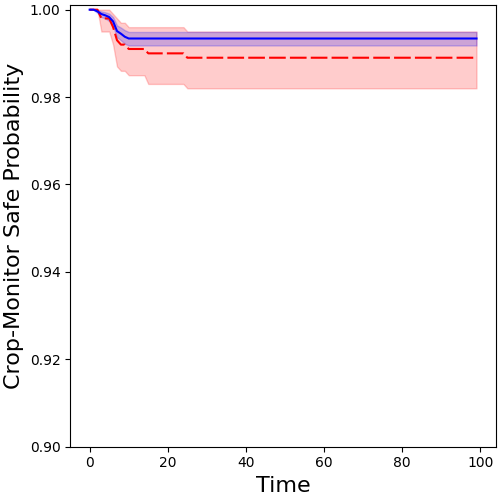}
    \caption{Probability of remaining in a safe state over time steps}
    \label{fig:ex:safeProb}
  \end{minipage}
\end{figure}

\paragraph{Accuracy.} Figure~\ref{fig:ex:cdf} shows a comparison of the heading
and distance distributions after 100 time steps. The X~Axis shows the variable
value and the Y~Axis shows the cumulative probability. The blue solid and red
dashed plots show the distributions estimated using MCS and \NAME{},
respectively. We compare the \NAME{} and MCS distributions using the
Kolmogorov-Smirnov (KS) statistic and the Wasserstein metric.
Figure~\ref{fig:ex:wass} shows how the Wasserstein metric evolves over time. The
X~Axis shows time steps and the Y~Axis shows the Wasserstein metric. The low
Wasserstein metric value indicates good correlation between the two
distributions at all times. The KS statistic also remains below 0.14.
Figure~\ref{fig:ex:safeProb} shows the probability of remaining in a safe state
over time. The X~Axis shows time steps and the Y~Axis shows the probability of
remaining safe. The shaded regions show the 95\% bootstrap confidence interval.
Around each plot, the bootstrap confidence interval indicates the variation that
can occur as a result of sampling error. As \NAME{} evaluates $M_\mathit{GPC}$
for $10\times$ more samples than $M_V$, its confidence interval is smaller. We
use the t-test to check if the safe state probabilities are similar: it passes
for 99 of 100 time \mbox{steps, indicating high similarity.}

\paragraph{Time.} \NAME{} is 2.3x faster than MCS on our hardware. MCS required
19.5~hours. Gathering the training data for the perception model required
8.5~hours. The time required for training the perception model, constructing the
polynomial approximation with GPC, and using the GPC approximation was
negligible in comparison ($<1$ minute). Increasing the number of samples or the
number of time steps would further increase the gap between the two methods,
since the time required to construct the perception model is a one-time cost.

\section{Background}
\label{sec:bkg}

We present key definitions pertaining to GPC. Dedicated books
(e.g.,~\cite{XiuGPCBook}), provide more details.

\paragraph{Orthogonal polynomials.} Assume $X$ is a continuous variable with
support $S_X$ and probability density $p_X : S_X \rightarrow \mathbb R$. Let
$\Psi = \{ \Psi_n | n \in \mathbb N \}$ be a set of polynomials, where $\Psi_n$
is an $n^{th}$ degree polynomial. Then $\Psi$ is a set of orthogonal polynomials
with respect to $X$ if for all $n \ne m$, $\int_{S_X} \Psi_n(x) \Psi_m(x) p_X(x)
dx = 0$. The orthogonal polynomial $\Psi_n$ has $n$ distinct roots in $S_X$.

Orthogonal polynomials exist for several probability distributions. For example,
the Legendre, Hermite, Jacobi, and Laguerre polynomials are orthogonal for the
uniform, normal, beta, and gamma distributions respectively.

\paragraph{Orthogonal polynomial projection (GPC).} Let $f : S_X \rightarrow
\mathbb R$. Then the $N^{th}$ order orthogonal polynomial projection of $f$,
written as $f_N$, with respect to a set of orthogonal polynomials $\Psi$, is:
\begin{equation}
\label{eq:fnApprox}
f_N = \sum\limits_{i=0}^{N} c_i \Psi_i
\quad\text{where}\quad
c_i = \frac{\int_{S_X} f(x) \Psi_i(x) p_X(x) dx}{\int_{S_X} \Psi_i^2(x) p_X(x) dx}
\end{equation}
If $f$ is an $N^{th}$ degree polynomial, then $f_N = f$. Otherwise, $f_N$ is the
\emph{optimal} $N^{th}$ degree polynomial approximation of $f$ w.r.t. $X$, in
the sense that it minimizes $\ell_2$ error, which is calculated as $\int_{S_X}
(f(x) - f_N(x))^2 p_X(x) dx$. As $N\rightarrow\infty$, the $\ell_2$ error
approaches 0, that is, we can construct arbitrarily good approximations of $f$.
$f_N$ is called the \textit{$N^{th}$-order generalized polynomial chaos (GPC)
approximation.}

\paragraph{Lagrange basis polynomials.} Given $N$ points $(x_i,y_i),\ 1 \le i
\le N$, where all $x_i$ are distinct, Equation~\ref{eq:bkg:lagrange} shows the
\emph{Lagrange basis polynomials} $L_i$ for each $i$.
\begin{equation}
    \label{eq:bkg:lagrange}
    L_i(x) = \prod\limits_{\substack{1 \le j \le N \\ j \ne i}} \frac{x - x_j}{x_i - x_j}
\end{equation}

\paragraph{Gaussian quadrature.} To use Equation~\ref{eq:fnApprox}, we must
perform multiple integrations to calculate the coefficients $c_i$ ($i\in\{0
\ldots N\}$). For any non-trivial function $g$, we must use numerical
integration by approximating the integral with the following sum:
\begin{equation}
\label{eq:bkg:quadSum}
\int\limits_{S_X} g(x) p_X(x) dx \approx \sum \limits_{i=1}^{N} w_i g(x_i)
\ \ \text{where}\ \
w_i = \int\limits_{S_X} L_i(x) p_X(x) dx
\end{equation}
We choose $w_i$ and $x_i$ so as to minimize integration error. In \emph{Gaussian
quadrature}, we choose $x_i$ to be the $N$ roots of $\Psi_N$, the $N^{th}$ order
orthogonal polynomial w.r.t. $X$. We calculate the corresponding weights using
the Lagrange basis polynomials $L_i$ (Equation~\ref{eq:bkg:lagrange}) passing
through $x_j\ \forall j \ne i$.

\newcommand{\mbsx}{\hspace{-.1em}\textbf{\em{x}}}

\paragraph{Multivariate GPC.} GPC can be easily extended to the multivariate
case, as long as all \emph{all random variables are independent}. Let $\mathbf X
= (X_1, \ldots, X_d)$ be the $d$ independent random variables (not necessarily
following the same distribution) and let $f$ be a function over $\mathbf X$. The
orthogonal polynomials $\mathbf{\Psi_i}$ for $\mathbf X$ are simply the products
of the orthogonal polynomials $\Psi_{i_1}, \ldots, \Psi_{i_d}$ for $X_1, \ldots,
X_d$ respectively. The GPC approximation closely resembles the one for the
univariate case:
\begin{equation}
\label{eq:bkg:MVGPC}
f_N = \sum\limits_{\mathbf i} c_{\mathbf i} \mathbf{\Psi_i}
\quad\text{where}\quad
c_{\mathbf i} = \frac{\int_{S_{\mathbf X}} f(\mbsx{}) \mathbf{\Psi_i}(\mbsx{}) p_{\mathbf X}(\mbsx{}) d\mbsx{}}
                     {\int_{S_{\mathbf X}} \mathbf{\Psi_i}^2(\mbsx{}) p_{\mathbf X}(\mbsx{}) d\mbsx{}}
\end{equation}
We calculate $c_{\mathbf i}$ using a variant of Equation~\ref{eq:bkg:quadSum} in
which we sum over all dimensions of $\mathbf{i}$.

\paragraph{Global sensitivity (Sobol) indices.} Sobol indices~\cite{sobolSens}
decompose the variance of the model output over the entire input distribution
into portions that depend on subsets of the input variables $X_i$. The
\emph{first order} sensitivity indices show the contribution of a single input
variable to the output variance. For a variable $X_i$, the sensitivity index is
$S_i=V_i/V$. Here, $V=\mathit{Var}_{\mathit{\mathbf X}}(f(\mbsx{}))$ \mbox{is
the total variance and}
\begin{align}
\label{eq:bkg:partvar}
V_i &= \mathit{Var}_{X_i}(E_{\mathit{\mathbf X_{\neg i}}}(f(\mbsx{}) | X_i = x_i)) \\
\text{where}\quad
\mathit{\mathbf X_{\neg i}} &= \{X_1, \ldots, X_d\} \setminus \{ X_i \} \nonumber
\end{align}
We can evaluate Equation~\ref{eq:bkg:partvar} analytically when $f$ is a
polynomial (such as those generated via GPC) and when it is possible to
calculate the moments of each independent component of $\mathbf X$ analytically.
For more complex functions and distributions, it becomes necessary to estimate
Equation~\ref{eq:bkg:partvar} empirically using Monte Carlo
estimators~\cite[Equation~6]{sobolSens}.
\section{\NAME{} Approach}
\label{sec:approach}

We present the \NAME{} approach for creating a polynomial surrogate model of
complex autonomous vehicle systems. \NAME{} consists of three high level steps:
\begin{enumerate}

\item Create a deterministic complete vehicle model.

\item Train a perception model and use it to replace the regression neural
network used for state perception in the vehicle model
(Algorithms~\ref{code:app:perTrain}-\ref{code:app:modVehMod}).

\item Construct a GPC surrogate model from the vehicle model
(Algorithm~\ref{code:app:gpcCons}).

\end{enumerate}
\NAME{} automates almost the entire process of constructing and using the
surrogate model. The user provides the perception model training data, the
distributions of state and random variables, GPC order, and simulation
parameters. \NAME{} infers the optimal degree of polynomial regression for the
perception model using cross-validation to maximize accuracy while preventing
overfitting. We plan to release \NAME{} as open-source software in the future.

\subsection{Creating a deterministic vehicle model}

First, we represent the complete vehicle model as a function of independent
random variables, $M_V : \mathbb{D}_S \times \mathbb{D}_E \times \mathbb{D}_R
\rightarrow \mathbb{D}_S$. $S \in \mathbb{D}_S$ is a vector of state variables
(e.g., position and rotation), $E \in \mathbb{D}_E$ is a vector of
environment-related random variables that affect the neural network (e.g.,
weather and lighting conditions that affect the image processed by the neural
network), and $R \in \mathbb{D}_R$ is a vector of random variables that do not
affect the neural network, but affect other parts of $M_V$. Making $M_V$
deterministic is necessary as GPC produces a deterministic polynomial model over
independent random variables. Because multivariate GPC requires the input
variables to be independent, we remove any input variable dependencies by
isolating their independent components and using those as the input variables
instead.

For the example in Section~\ref{sec:ex}, $\mathbb{D}_S$ is the Cartesian product
of the range of the heading and distance variables, $\mathbb{D}_E$ has variables
that control crop type, crop age, and lighting, and $\mathbb{D}_R$ is unused.
The top half of Figure~\ref{fig:ex:blockDiag} describes $M_V$.

\subsection{Replacing the perception system}
\label{app:permod}

The output of regression neural networks which use camera images to perceive the
vehicle's state is affected by environmental factors unrelated to the vehicle's
state. For a given ground truth state, such perception systems produce a
distribution of possible perceived states. To enable faster sampling of this
output distribution, \NAME{} replaces the perception system with a perception
model prior to constructing the GPC polynomial model.
Algorithm~\ref{code:app:perTrain} shows how \NAME{} creates the perception
model. The steps are as follows:
\begin{enumerate}

\item Let $\mathbb{D}_S^\textit{safe}$ be the set of safe states. We pick a set
of ground truth states $G \subset \mathbb{D}_S^\textit{safe}$, such that $G$ is
an evenly spaced tensor grid that includes the extreme states within
$\mathbb{D}_S^\textit{safe}$.

\item For each $g \in G$, \NAME{} captures a list of images $I_g$ using a
photo-realistic simulator such as Gazebo~\cite{gazeboSim} or
CARLA~\cite{Dosovitskiy17}. For each image, \NAME{} randomly samples $E$ from
the environment distribution $\mathcal{D}_E$ ($E \sim \mathcal{D}_E$). \NAME{}
obtains a large number of images per state ($N_i \ge 350$) to get a high
confidence estimate of the output distributions in the next step.

\item For each $g \in G$, \NAME{} passes $I_g$ through the neural network to
obtain a list of network outputs $O_g$. \NAME{} calculates the mean $\mu_g$ and
covariance $\sigma^2_g$~of~$O_g$.

\item \NAME{} trains a polynomial regression model $M_\mathit{per}$ to predict
the components of $\mu_S$ and $\sigma_S^2$ given $S \in
\mathbb{D}_S^\textit{safe}$.

\end{enumerate}

$M_\mathit{per}$ returns the \emph{distribution parameters} $(\mu_S,\sigma_S^2)$
of possible outputs. We create an \emph{abstracted} vehicle model $M'_V :
\mathbb{D}_S \times \mathbb R^n \times \mathbb{D}_R \rightarrow \mathbb{D}_S$
(Algorithm~\ref{code:app:modVehMod}) to use $M_\mathit{per}$ instead of the
neural network. Instead of a sample from $\mathcal{D}_E$, $M'_V$ accepts a
sample $N$ from a multivariate standard normal distribution $\mathcal N(0, 1)$.
It calculates $\mu_S$ and $\sigma_S^2$ using $M_\mathit{per}$, and transforms
$N$ to a sample from $\mathcal N(\mu_S, \sigma_S)$. $M'_V$ uses this sample as
the perceived state for the rest of the model consisting of the vehicle's
control and dynamics systems.

For the example in Section~\ref{sec:ex}, $\mathbb{D}_S^\textit{safe}$ is the
Cartesian product of the safe range of the heading and distance variables, $G$
is an $11 \times 11$ grid within $\mathbb{D}_S^\textit{safe}$, and
$M_\mathit{per}$ is a degree 4 polynomial regression model that predicts the
five parameters of the distribution of the perceived state (mean heading and
distance, heading and distance variance, and correlation). The bottom half of
Figure~\ref{fig:ex:blockDiag} describes $M'_V$. As there are two state
variables, the raw sample is drawn from a 2D standard distribution.

\begin{algorithm}[H]\algsize
    \begin{algorithmic}[1]
      \Input{$G$: set of ground truth states;
      $\mathcal{D}_E$: distribution of environment variables;
      $N_i$: number of images to capture for each $g \in G$}
      \Returns{$M_\mathit{per}$: trained perception model;
      $d_\mathit{per}$: polynomial degree of perception model}
      \Function{TrainPerceptionModel}{$G, \mathcal{D}_E, N_i$}
      \State $\mathit{TrainTestData} \gets \{\;\}$
      \For{$g \in G$}
        \State $I_g \gets [\;]$
        \For{$i$ from $1$ to $N_i$}
          \State $E \sim \mathcal{D}_E$
          \State $\mathit{Img} \gets \Call{CaptureImage}{g, E}$
          \State $I_g \gets I_g :: \mathit{Img}$
        \EndFor
        \State $O_g \gets \Call{NeuralNetwork}{I_g}$
        \State $\mu_g \gets \Call{Mean}{O_g}$
        \State $\sigma^2_g \gets \Call{Covariance}{O_g}$
        \State $\mathit{TrainTestData} \gets \mathit{TrainTestData}[g \mapsto (\mu_g, \sigma^2_g)]$
      \EndFor
      \State $M_\mathit{per}, d_\mathit{per} \gets \Call{PolyRegressionModel}{\mathit{TrainTestData}}$
      \EndFunction
    \end{algorithmic}
    \caption{Training the perception model}
    \label{code:app:perTrain}
  \end{algorithm}

\begin{algorithm}[H]\algsize
    \begin{algorithmic}[1]
      \Input{$S$: initial state of vehicle;
      $N$: raw sample to be transformed into neural network output sample;
      $R$: other random variables;
      $M_\mathit{per}$: trained perception model}
      \Returns{$S'$: state of vehicle after one time step}
      \Function{$M'_V$}{$S, N, R, M_\mathit{per}$}
        \State $\mu_S, \sigma^2_S \gets M_\mathit{per}(S)$
        \State $O_S \gets \Call{Transform}{N, \mu_S, \sigma^2_S}$ \label{line:app:sampTransform}
        \State $S' \gets \Call{VehicleControlAndDynamics}{S, O_S,  R}$ \label{line:app:ctrlDynCall}
      \EndFunction
    \end{algorithmic}
    \caption{Abstracted vehicle model}
    \label{code:app:modVehMod}
  \end{algorithm}

We assume the output is distributed according to $\mathcal N(\mu_S, \sigma_S)$
as we can easily define the distribution in terms of $\mu_S, \sigma_S^2$. We
have observed that, for real-world input images, the output distribution of
perception neural networks indeed tends to be normally distributed
(Figure~\ref{fig:ex:agbotrealerr}). However, we can use the same method for
other distributions if 1)~the parameters of the fitted distribution vary
smoothly as the ground truth changes, and 2)~orthogonal polynomials
corresponding to that distribution exist for use with GPC (e.g., the uniform and
beta distributions).

\subsection{GPC for the complete vehicle system}
\label{sec:app:gpccons}

Algorithm~\ref{code:app:gpcCons} shows how \NAME{} constructs the GPC
approximation of the abstracted vehicle model:
\begin{enumerate}

\item \NAME{} constructs a joint distribution $J$ over $\mathbb{D}_S \times
\mathbb R^n \times \mathbb{D}_R$. For the state variables, \NAME{} chooses a
normal or truncated normal distribution $\mathcal{D}_S^\textit{safe}$ over the
safe state space $\mathbb{D}_S^\textit{safe}$. \NAME{} uses $\mathcal N(0, 1)$
for the perception model raw sample. For the other random variables, \NAME{}
uses their actual distribution~$\mathcal{D}_R$.

\item \NAME{} constructs the basis polynomials which are orthogonal w.r.t. $J$.
Increasing the order $o_\mathit{gpc}$ increases accuracy, but also runtime.

\item \NAME{} chooses Gaussian quadrature nodes for $J$ and calculates the
corresponding weights using Equation~\ref{eq:bkg:quadSum}.

\item \NAME{} evaluates the abstracted model $M'_V$ at the chosen samples.

\item \NAME{} calculates the GPC polynomial model as the weighted sum of the
orthogonal polynomials using Equation~\ref{eq:fnApprox}.

\end{enumerate}

For the example in Section~\ref{sec:ex}, $J=\mathcal N(0, \pi/18) \times
\mathcal N(0, 0.076) \times \mathcal N(0, 1)^2$. Specifically, for the state
variables, we choose normal distributions such that the edge of the safe state
space lies at the $3\sigma$ boundary. $\Psi$ is a set of 70 multivariate Hermite
polynomials (the orthogonal polynomials for normal distributions). These
polynomials are formed by taking the product of the Hermite polynomials for each
variable in $J$, such that the total order is at most 4. $X$ and $W$ are a set
of 624 quadrature nodes and the corresponding weights. Lastly, $M_\mathit{GPC}$
is an order 4 polynomial over the variables in $J$.

\begin{algorithm}[H]\algsize
    \begin{algorithmic}[1]
      \Input{$\mathcal{D}_S^\textit{safe}$: distribution over $\mathbb{D}_S^\textit{safe}$;
      $\mathcal{D}_R$: distribution of other random variables;
      $o_\mathit{gpc}$: order of GPC model;
      $M'_V$: abstracted vehicle model}
      \Returns{$M_\mathit{GPC}$: \NAME{} surrogate model}
      \Function{CreateGPCModel}{$\mathcal{D}_S^\textit{safe}, \mathcal{D}_R, o_\mathit{gpc}, M'_V$}
        \State $J \gets \Call{Join}{\mathcal{D}_S^\textit{safe}, \mathcal N(0, 1), \mathcal{D}_R}$
        \State $\Psi \gets \Call{GenerateOrthogonalPolynomials}{o_\mathit{gpc}, J}$
        \State $X, W \gets \Call{GenerateQuadNodesAndWeights}{o_\mathit{gpc}, J}$
        \State $Y \gets [M'_V(x)$ for $x \in X]$
        \State $M_\mathit{GPC} \gets \Call{QuadratureAndGPC}{\Psi, X, W, Y}$
      \EndFunction
    \end{algorithmic}
    \caption{\NAME{} surrogate model construction}
    \label{code:app:gpcCons}
  \end{algorithm}

\paragraph{Categorical state variables.} Some vehicle models have categorical
state variables. For example, many control systems operate in multiple modes.
The control system can switch modes if certain conditions are met, and the
current mode affects the control decisions. In this case, the current mode is a
categorical state variable.

Unlike categorical variables, polynomial inputs and outputs are continuous
intervals. Therefore, we cannot use GPC for predicting categorical variables, or
accept a categorical variable as an input to the GPC model. \NAME{} solves this
problem by using multiple GPC sub-models and using a separate classifier for
predicting categorical variables. This procedure is known as \emph{multi-element
GPC} (ME-GPC).

Consider a vehicle model $M_V$ whose state includes a categorical variable $X$
with the domain $\mathbb D_X = \{x_1, \ldots, x_k\}$. \NAME{} uses GPC to create
a separate polynomial model for each $x_i \in \mathbb D_X$. The compound
surrogate model chooses which of these sub-models to use based on the current
vehicle state. In this way, \NAME{} calculates all output state variables
except~$X$. For predicting $X$, \NAME{} creates an ancillary classifier, which
it trains as follows:
\begin{enumerate}

\item \NAME{} picks a set of states $G$ from the safe states
$\mathbb{D}_S^\textit{safe}$, such that $G$ is an evenly spaced tensor grid that
includes the most extreme states within $\mathbb{D}_S^\textit{safe}$, and
includes all categorical values in $\mathbb D_X$.

\item For each $g \in G$, \NAME{} evaluates $M_V(g, E, R)$ for a large number of
samples ($\ge 350$) of $E$ and $R$. \NAME{} isolates the value of $X$ in the
output and determines the most frequent value $X_\mathit{mode}$.

\item \NAME{} trains a classifier to predict $X_\mathit{mode}$ for any $S \in
\mathbb{D}_S^\textit{safe}$.

\end{enumerate}
\NAME{} returns the value $X_\mathit{mode}$ predicted by the classifier as the
value of $X$ in the output. This method can be generalized to multiple
categorical variables, but the number of separate GPC approximations required
can become impractical.

\subsection{Applications of the \NAME{} surrogate model}

\paragraph{Calculating probability of remaining in a safe state over time.}
\NAME{} uses the surrogate model to estimate the probability that the vehicle
will remain in a safe state over time (Algorithm~\ref{code:app:gpcUseApprox}).
\NAME{} creates initial joint samples using the initial state distribution, the
raw sample distribution, and the distribution of other random variables. At each
time step, \NAME{} evaluates the surrogate model on each joint sample to get the
next state. If the next state is safe, \NAME{} chooses new random values to
prepare the joint sample for the next time step. Finally, \NAME{} calculates and
logs the fraction of samples that are still in the safe region.

\begin{algorithm}[H]\algsize
    \begin{algorithmic}[1]
      \Input{$N_s$: number of samples to use for distribution estimation;
      $T$: number of time steps;
      $\mathcal{D}_S^0$: initial state distribution;
      $\mathcal{D}_R$: distribution of other random variables;
      $\mathit{Pred}$: safety predicate;
      $M_\mathit{GPC}$: constructed \NAME{} surrogate model}
      \Returns{$P_\mathit{safe}$: probability that the vehicle is safe until each time step}
      \Function{EstimateSafeProb}{$N_s, T, \mathcal{D}_S^0, \mathcal{D}_R, \mathit{Pred}, M_\mathit{GPC}$}
        \State $X \gets [\;]$; $P_\mathit{safe} \gets \{\;\}$
        \For{$i$ from $1$ to $N_s$}
          \State $x \sim \Call{Join}{\mathcal{D}_S^0, \mathcal N(0, 1), \mathcal{D}_R}$ \label{line:app:sampInit}
          \State $X \gets X :: x$
        \EndFor
        \For{$t$ from $1$ to $T$}
          \State $X' \gets [\;]$
          \For{$x \in X$}
            \State $S' \gets M_\mathit{GPC}(x)$
            \If{$\Call{Safe}{S',\mathit{Pred}}$}
              \State $N' \sim \mathcal N(0, 1)$
              \State $R' \sim \mathcal{D}_R$
              \State $X' \gets X' :: (S', N', R')$
            \EndIf
          \EndFor
          \State $X \gets X'$
          \State $P_\mathit{safe} \gets P_\mathit{safe}[t \mapsto |X| / N_s]$
        \EndFor
      \EndFunction
    \end{algorithmic}
    \caption{Estimating the probability of remaining in a safe state over time}
    \label{code:app:gpcUseApprox}
  \end{algorithm}

\begin{algorithm}[H]\algsize
    \begin{algorithmic}[1]
      \Input{$N_s$: number of samples to use for sensitivity estimation;
      $i$: index of state variable to calculate sensitivity for;
      $\mathcal{D}_S$: current state distribution;
      $\mathcal{D}_R$: distribution of other random variables;
      $M$: model ($M_\mathit{GPC}$ or $M'_V$)}
      \Returns{$S_i$: sensitivity index of selected state variable}
      \Function{EstimateSensitivity}{$N_s, i, \mathcal{D}_S, \mathcal{D}_R, M$}
        \State $Y_0 \gets [\;]$; $Y_1 \gets [\;]$
        \For{$i$ from $1$ to $N_s$}
          \State $x_0, x_1 \sim \Call{Join}{\mathcal{D}_S, \mathcal N(0, 1), \mathcal{D}_R}$
          \State $y_0 \gets M(x_0)$; $y_1 \gets M(x_1[i \mapsto x_0[i]])$
          \State $Y_0 \gets Y_0 :: y_0$; $Y_1 \gets Y_1 :: y_1$
        \EndFor
        \State $S_i \gets (\Call{Mean}{Y_0 * Y_1} - \Call{Mean}{Y_0}^2) / \Call{Var}{Y_0}$
      \EndFunction
    \end{algorithmic}
    \caption{Using estimators to calculate sensitivity indices}
    \label{code:app:mcSens}
  \end{algorithm}

\paragraph{Computing Sobol indices.} \NAME{} uses $M_\mathit{GPC}$ to calculate
Sobol sensitivity indices in two ways. In the \emph{analytical} approach,
\NAME{} calculates sensitivity indices by first calculating conditional expected
values as polynomials and then calculating their variance
(Equation~\ref{eq:bkg:partvar}). In the \emph{empirical} approach, \NAME{}
instead uses Monte Carlo estimators (\cite[Equation~6]{sobolSens} as implemented
in Algorithm~\ref{code:app:mcSens}).

While the analytical approach precisely calculates sensitivity indices, it is
relatively slow as \NAME{} must calculate expected values as a function of the
variable whose sensitivity is being calculated. The empirical approach becomes
more accurate as the number of samples increases. Despite this, it can be faster
than the first approach due to the speed of evaluating $M_\mathit{GPC}$.

\paragraph{Rapid iteration.} During development of an autonomous vehicle system,
the vehicle model can change rapidly as the perception neural network or vehicle
control parameters are tweaked. \NAME{} enables faster testing of these
prototypes thanks to its compositional approach. If the same perception system
is used while changing the control and dynamics, then \NAME{} saves time by
reusing the existing perception model. If the perception system is changed,
\NAME{} must rerun
\mbox{Algorithms~\ref{code:app:perTrain}-\ref{code:app:gpcCons}}, but if doing
so is faster than using MCS, the time saved adds up with each iterative change.

\subsection{Properties of the \NAME{} approach}
\label{sec:app:prop}

\paragraph{Accuracy.} Multiple \NAME{} parameters affect the accuracy of the GPC
model: the size of the tensor grid $|G|$, the number of images taken for each
grid point $N_i$, the degree of polynomial regression used for the perception
model, and the GPC order $o_\mathit{gpc}$. Under certain conditions, \NAME{}
converges in distribution to the exact solutions.

\renewcommand*{\proofname}{Proof Sketch}

\begin{lemma}[Perception Model Convergence]
\label{lem:perModConv}
Assume that 1)~for the given environment distribution $\mathcal{D}_E$, the
distribution of the outputs of a perception neural network $\mathcal{N}$ in any
ground truth state $S$ is Gaussian over the perceived state, and 2)~each
component of the distribution parameters $(\mu_S,\sigma_S^2)$ is an analytic
function of $S$. Then, the output distribution of the perception model
$M_\mathit{per}$ in any state approaches $\mathcal{N}$'s output distribution at
that state as $|G|$, $N_i$, and $d_\mathit{per}$ increase.
\end{lemma}
\begin{proof}
Increasing $|G|$ increases the number of ground truth states used to train the
perception model. Increasing $N_i$ increases the accuracy of $\mathcal{N}$'s
output distribution parameters calculated at each $S$. Since these distribution
parameters are analytic functions of $S$, they can be calculated using a Taylor
series over $S$. After increasing the number and accuracy of training data
points, the accuracy of the perception model can be arbitrarily increased by
increasing $d_\mathit{per}$\footnote{Increasing $d_\mathit{per}$ without also
increasing $|G|$ leads to overfitting.}.
\end{proof}

We use standard statistical tests such as the Shapiro-Wilk test to check if
$\mathcal{N}$'s outputs have a Gaussian distribution for the environment
distribution $\mathcal{D}_E$ used in Section~\ref{app:permod}. We have observed
this to be true in practice (Figure~\ref{fig:ex:agbotrealerr}). We can also use
a different base distribution (and corresponding orthogonal polynomials for GPC)
if fits the data better across the state space.

Practically, controlling the error of the perception model (or any approximation
of a neural network) is an open
problem~\cite{regressionAltNN,decomposingNNModules}. Precise analytic
calculation of the perception model error is intractable, but we can empirically
estimate the error.

\begin{lemma}[GPC Error Bound]
\label{lem:gpcErr}
Assume the control system and vehicle dynamics in $M'_V$ are differentiable.
Then, the root mean square (RMS) error of the output of the \NAME{} model
$M_\mathit{GPC}$ with respect to the output of $M'_V$ is bounded.
\end{lemma}
\begin{proof}
From \cite[Theorem~3.6]{XiuGPCBook} and Ernst et
al.~\cite{ernst2012convergence}, which state that the RMS error of a GPC
approximation is proportional to $o_\mathit{gpc}^{-p}$, where $p$ is a positive
value that depends on the differentiability of the function being approximated.
The process of generating a neural network output sample through the perception
model is a polynomial evaluation followed by an affine transform -- both are
differentiable operations. The control system and dynamics are differentiable by
assumption. Finally, composing differentiable functions yields a differentiable
function.
\end{proof}

\noindent
$M_\mathit{GPC}$ is the optimal polynomial model of $M'_V$ for any
$o_\mathit{gpc}$ (\cite{XiuGPCBook}[Equation~5.9]). In practice, control systems
may not be differentiable everywhere (e.g., due to mode switching), but the
differentiability of vehicle dynamics, coupled with a short duration time step,
limit negative effects on accuracy.

\begin{corollary}[GPC Convergence]
\label{cor:gpcConverge}
As $o_\mathit{gpc} \to \infty$, RMS error of GPC approaches 0, that is,
$M_\mathit{GPC}$ can be an arbitrarily close approximation of $M'_V$.
\end{corollary}
\begin{proof}
From Lemma~\ref{lem:gpcErr}, the RMS error is proportional to
$o_\mathit{gpc}^{-p}$, where $p$ is positive. Then, $\lim\limits_{o_\mathit{gpc} \to
\infty} o_\mathit{gpc}^{-p} = 0$.
\end{proof}

\begin{theorem}[\NAME{} Convergence]
Assume that the distribution of the outputs of a perception neural network
$\mathcal{N}$ in any ground truth state is Gaussian. Then, the \NAME{} model
$M_\mathit{GPC}$ converges in output distribution to the original vehicle model~$M_V$.
\end{theorem}
\begin{proof}
Since we can construct an arbitrarily accurate perception model
(Lemma~\ref{lem:perModConv}), we can use it to obtain accurate neural network
output samples for any state in $M'_V$. The GPC model can be made an arbitrarily
accurate approximation of $M'_V$ (Corollary~\ref{cor:gpcConverge}), and thus of
$M_V$.
\end{proof}

\renewcommand*{\proofname}{Proof}

\paragraph{Runtime.} The dominant factor for runtime is the required number of
evaluations of $M_V$. To gather data for the perception model, \NAME{} requires
$\Theta(|G| N_i)$ evaluations of $M_V$. The amount of time required to train
$M_\mathit{per}$ and construct $M_\mathit{GPC}$ is insignificant in comparison.

For state distribution estimation over time, MCS requires $\Theta(N_s T)$
evaluations of $M_V$ ($N_s$ being the number of samples used for distribution
estimation), while \NAME{} requires the same number of evaluations of the much
faster $M_\mathit{GPC}$. For estimating sensitivity indices using estimators, we
must evaluate either $M'_V$ or $M_\mathit{GPC}$ $\Theta(N_s)$ times, followed by
mean and variance calculations.
\begin{table}
\centering
\small
\caption{\NAME{} benchmarks}
\label{tab:bench}
\begin{tabular}{@{\hskip 0em}l@{\hskip .7em}l@{\hskip .7em}l@{\hskip .7em}r@{\hskip .1em}c@{\hskip .1em}l@{\hskip .7em}c@{\hskip 0em}}
\textbf{Benchmark} & \textbf{Perception} & \textbf{Control} & \multicolumn{3}{l}{\textbf{Replacement}}
& $\mathbf{dim_{s/r}}$\\

\midrule

Crop-Monitor & ResNet-18$\times2$ & Skid-Steer & Perc & $\rightarrow$ & Poly Reg & 2/2 \\

Car-Straight & LaneNet & Pure Pursuit & Perc & $\rightarrow$ & Poly Reg & 2/2 \\

Car-Curved & LaneNet & Pure Pursuit & Perc & $\rightarrow$ & Poly Reg & 2/2 \\

ACAS-Table & Ground Truth & ACAS-Xu Table & Ctrl & $\rightarrow$ & Dec Tree & 4/0 \\

ACAS-NN & Ground Truth & ACAS-Xu NN & Ctrl & $\rightarrow$ & Dec Tree & 4/0 \\

\end{tabular}
\end{table}

\section{Methodology}
\label{sec:methodology}

\paragraph{Benchmarks.} We chose five benchmarks that include autonomous vehicle
systems such as self driving cars, unmanned aircraft, and crop monitoring vehicles.
Table~\ref{tab:bench} shows details of the benchmarks. Columns~2~and~3 state the
vehicle's perception and control system, respectively. Column~4 indicates if
\NAME{} made a replacement in the perception (Perc) or control (Ctrl) system,
and the nature of the replacement (Poly Reg: polynomial regression, Dec Tree:
decision tree). Column~5 states the number of state and random variable
dimensions. The benchmarks are:
\begin{itemize}

\item \textbf{Crop Monitoring Vehicle.} A vehicle that travels between two rows
of crops and must avoid hitting them. This is our main example
(Section~\ref{sec:ex}).

\item \textbf{Self-Driving Car on a Straight Road.} A vehicle that must drive
within a road lane ($\mathbb{D}_S^\textit{safe} \equiv |\mathit{heading}| \le
\pi/12 \wedge |\mathit{distance}| \le 1.2m$). It uses LaneNet to perceive the
lane boundaries and uses the pure pursuit controller. We derive this benchmark
from~\cite{9294366} and use~\cite{lanenet-lane-detection}.

\item \textbf{Self-Driving Car on a Curved Road.} Similar to the previous
benchmark, but the vehicle must drive on a circular road of radius 100m.

\item \textbf{Unmanned Aircraft Collision Avoidance (Lookup Table).} An unmanned
aircraft that must avoid a near miss with an intruder
($\mathbb{D}_S^\textit{safe} \equiv |\mathit{separation}| \ge 0.1524 km$). The
aircraft uses ACAS-Xu lookup tables from~\cite{horizontalCAS}. As this model's
state includes a categorical variable (the previous ACAS advisory), we use
ME-GPC and predict the next advisory using a decision tree as the ancillary
model.

\item \textbf{Unmanned Aircraft Collision Avoidance (Neural Network).} Similar
to the previous benchmark, but uses a neural network from~\cite{horizontalCAS}
trained to replace the lookup table.

\end{itemize}

\paragraph{Implementation and experimental setup.} We performed our experiments
on machines with a Quadro P5000 GPU, using a single Xeon CPU core. We implement
\NAME{} in Python, using the \texttt{chaospy} library~\cite{chaospy}. We use
Gazebo 11~\cite{gazeboSim} to capture images for Crop-Monitor, Car-Straight, and
Car-Curved benchmarks. We run all image processing neural networks on the GPU.
We run ACAS-NN entirely on CPU as its network is small.

For estimating state distribution over time, we compare the \NAME{}-generated
$M_\mathit{GPC}$ to a MCS baseline using $M_V$. We set \NAME{} parameters as
follows: $G$ is a $11 \times 11$ grid in the safe state space, $N_i = 350$, and
$o_\mathit{gpc}=4$. We additionally experiment with alternate values for $G$ and
$N_i$, as they directly affect perception model training data generation time.
We set the number of time steps $T = 100$. To keep MCS runtime within 24 hours,
We set $N_s = 1,000$ for MCS. For \NAME{}, we increase $N_s$ to 10,000 as
$M_\mathit{GPC}$ is much faster than $M_V$ and increasing the number of samples
decreases sampling error for $M_\mathit{GPC}$.

For calculating sensitivity indices, we calculate sensitivity using both the
analytical and empirical method described in Section~\ref{sec:approach}. It is
not possible to compare sensitivity indices against $M_V$, as $M_V$ has a
different set of inputs (environment specification instead of a sample from
$\mathcal N(0, 1)$). Therefore, we use sensitivity index calculation using
$M'_V$ as the baseline. For the empirical method, we set $N_s = 10^6$, but also
monitor the results obtained by setting $N_s$ to $10^4$, $10^5$, and $10^7$. We
calculate the sensitivity of state variables to those in the previous time step,
as well as the sensitivity of the change in the state variables.

\paragraph{Environmental factors.} For the Crop-Monitor benchmark, our test
scenario includes two types of crops (corn and tobacco). There are four
different corn growth stages. For each type and growth stage, there are multiple
crop models (30 in total). For the Car benchmarks, our scenario includes the
presence of other cars, pedestrians, and skid marks that may obstruct lane
markings. We also vary lighting conditions.

\paragraph{Distribution similarity metrics.} We use two complementary similarity
metrics to separately compare each dimension of the MCS and \NAME{} state
distributions at each time step. The conservative \emph{KS statistic} quantifies
the maximum distance between the cumulative distribution functions of the two
distributions at any point. The \emph{Wasserstein metric} quantifies the minimum
probability mass that must be moved to transform one distribution into the
other. For both metrics, a lower value indicates greater distribution
similarity. We can use these distribution similarity metrics despite using more
samples for \NAME{} than for MCS.

We also compare the fraction of simulated vehicles remaining in the safe region
till each time step using three metrics. The two sample t-test is a statistical
test to check if the underlying distributions used to draw two sets of samples
are the same. The $\ell_2$ error is the RMS of the differences in safe state
probability at each time step. Lastly, we calculate the Pearson
cross-correlation coefficient between the two sets of safe state probabilities.
When plotting safe state probability, we also draw the 95\% bootstrap confidence
interval. This confidence interval does not directly compare the two plots, but
rather, for each individual plot, it provides an estimate of the variation that
can occur in that plot as a result of sampling error.
\section{Evaluation}
\label{sec:eval}

\begin{table}
      \centering
      \small
      \caption{Metrics for comparing state variable distributions}
      \label{tab:eval:kswass}
      \begin{tabular}{@{\hskip 0em}llr@{\hskip 0em}lr@{\hskip 0em}lrr@{\hskip 0em}}

      \textbf{Benchmark} & \textbf{Variable}
      & $\mathbf{\mu_{\NAME{}}}$ & $\mathbf{/\mu_{MCS}}$ & $\mathbf{\sigma_{\NAME{}}}$ & $\mathbf{/\sigma_{MCS}}$
      & $\mathbf{KS_{max}}$ & $\mathbf{Wass_{max}}$ \\

      \midrule

      \multirow{2}{*}{Crop-Monitor} & Heading (rad) & -0.004  & /-0.008 & 0.04 & /0.04 & 0.11 & 0.02 \\
                                    & Distance (m)  & -0.01   & /-0.02  & 0.03 & /0.03 & 0.14 & 0.01 \\

      \midrule

      \multirow{2}{*}{Car-Straight} & Heading (rad) & 0.0004 & /-0.0001 & 0.02 & /0.02 & 0.13 & 0.009 \\
                                    & Distance (m)  & 0.16   & /0.08    & 0.09 & /0.12 & 0.41 & 0.08 \\

      \midrule

      \multirow{2}{*}{Car-Curved} & Heading (rad) & -0.003 & /-0.005 & 0.006 & /0.007 & 0.17 & 0.003 \\
                                  & Distance (m)  & 0.18   & /0.19   & 0.04  & /0.05  & 0.15 & 0.02  \\

      \midrule

      \multirow{3}{*}{ACAS-Table} & Crossrange (km) & -0.07 & /-0.03 & 0.85 & /0.88 & 0.05 & 0.04 \\
                                  & Downrange (km)  & -0.61 & /-0.53 & 0.23 & /0.22 & 0.13 & 0.08 \\
                                  & Heading (rad)   & 0.63  & /-0.54 & 2.82 & /2.80 & 0.23 & 1.23 \\

      \midrule

      \multirow{3}{*}{ACAS-NN} & Crossrange (km) & -0.01 & /-0.01 & 0.93 & /0.91 & 0.02 & 0.03 \\
                               & Downrange (km)  & -0.45 & /-0.58 & 0.17 & /0.11 & 0.31 & 0.13 \\
                               & Heading (rad)   & 0.42  & /0.15  & 2.70 & /2.75 & 0.06 & 0.27 \\

      \end{tabular}
\end{table}

\begin{figure}
      \centering
      \includegraphics[width=.2\textwidth]{images/eval-safe/agbot.png}\hfill
      \includegraphics[width=.2\textwidth]{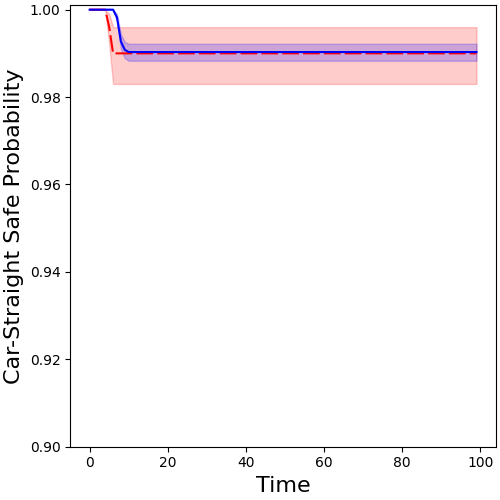}\hfill
      \includegraphics[width=.2\textwidth]{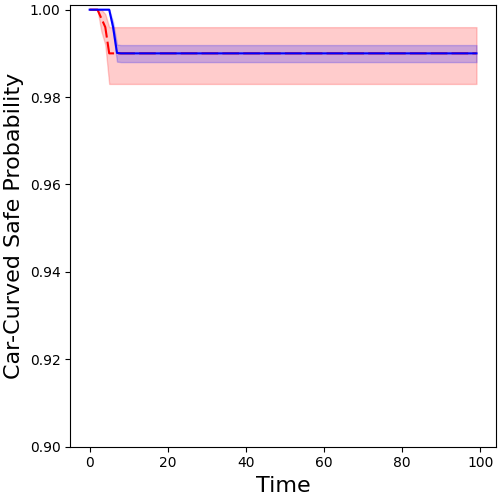}\hfill
      \includegraphics[width=.2\textwidth]{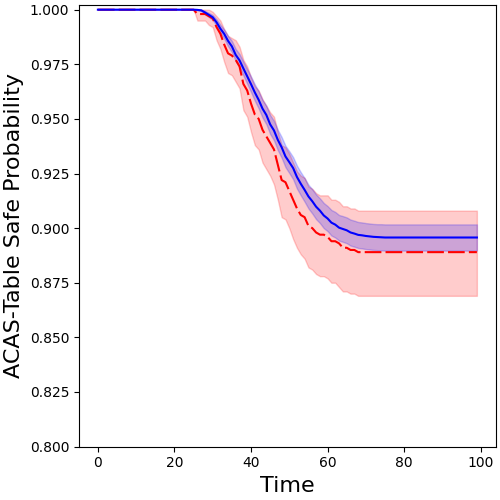}\hfill
      \includegraphics[width=.2\textwidth]{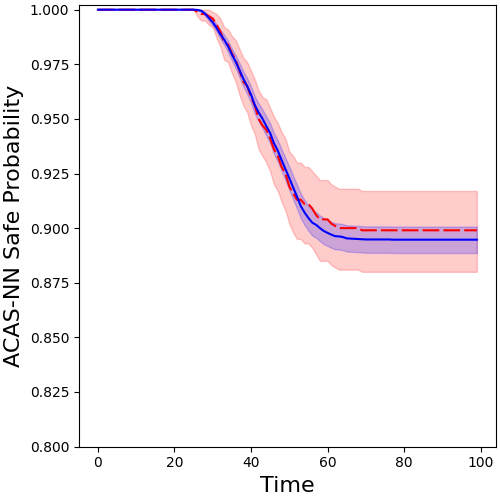}
      \caption{Evolution of safe state probability over time. Blue solid: \NAME{}, red dashed: MCS.}
      \label{fig:eval:safeProb}
\end{figure}

\subsection{Does \NAME{} accurately estimate the probability that a vehicle will
remain in a safe state over time?}
\label{sec:eval:stateest}

Table~\ref{tab:eval:kswass} compares the distributions calculated by \NAME{} and
MCS for each benchmark state variable. Columns~3-4 compare the mean and standard
deviation of the distributions at the \emph{final} time step. Columns~5-6 show
the maximum values of the KS statistic and Wasserstein metric over \emph{all}
time steps. For most state variables, the mean and standard deviation of the
distributions match closely up to the final time step. This is also indicated by
the low values of the Wasserstein metric and the conservative KS statistic. The
largest difference is for the Car-Straight benchmark distance distribution. This
occurs because the \NAME{} model and the original vehicle model converge towards
slightly different states around the center of the safe state space in later
time steps. However, during the initial time steps where more simulated vehicles
are in danger of entering unsafe states, the KS statistic does not exceed 0.15.
A similar phenomenon affects the ACAS-NN downrange distance variable. For
ACAS-Table, the \NAME{} and original vehicle models occasionally turn in
different directions to avoid an intruder approaching head-on, in situations
where turning in either direction is equally beneficial. This leads to a large
deviation in the heading variable.

Figure~\ref{fig:eval:safeProb} shows how the probability that the vehicle
remains in a safe state evolves over time. The blue solid and red dashed plots
show the probability estimates obtained using \NAME{} and MCS, respectively. The
shaded region around each plot shows the 95\% bootstrap confidence interval.
Because we use $10\times$ more samples when estimating safe state probability
with \NAME{} as compared to MCS, the sampling error is smaller for \NAME{},
which leads to a smaller confidence interval. Since \NAME{} approximates the
behavior of the vehicle model, as opposed to under/over approximation of
reachable states, \NAME{}'s safe state probability estimate can be on either
side of the MCS estimate.

Table~\ref{tab:eval:ssprob} shows the metrics we use to measure the similarity
of the safe state probabilities from Figure~\ref{fig:eval:safeProb}. Column~2
shows the number of time steps for which the t-test ``passed'', meaning that we
could not reject the null hypothesis that the probabilities are equal. Column~3
shows the $\ell_2$ error, and Column~4 shows the cross-correlation. The
similarity of the state distributions leads directly to the similarity of the
safe state probability for most time steps.

We extended the Crop-Monitor and Car benchmark experiments to 500 time steps to
confirm that the safe state probability does not deviate after 100 time steps.
We did not similarly extend the ACAS experiments as the ACAS system is primarily
relevant as the aircraft approach each other.

\textit{In conclusion, the vehicle state distributions that we estimate using
\NAME{} closely resemble those that we estimate using MCS, even after 100 time
steps. Consequently, the vehicle safe state probability estimate is also similar
between the two approaches.}

\begin{table}
\centering
\small
\begin{minipage}[t]{0.44\textwidth}
  \centering
  \caption{Metrics for comparing the probability of remaining in a safe state}
  \label{tab:eval:ssprob}
  \begin{tabular}{@{\hskip 0em}lrlr@{\hskip 0em}}
  \textbf{Benchmark} & \textbf{t-test} & \textbf{$\ell_2$ err} & \textbf{X-Cor} \\
  \midrule
  Crop-Monitor &  99/100 & 0.004 & 0.974 \\
  Car-Straight &  97/100 & 0.001 & 0.862 \\
  Car-Curved   &  98/100 & 0.001 & 0.865 \\
  ACAS-Table   & 100/100 & 0.007 & 0.998 \\
  ACAS-NN      & 100/100 & 0.003 & 0.999 \\
  \end{tabular}
\end{minipage}\hfill
\begin{minipage}[t]{.41\textwidth}
  \centering
  \caption{Maximum difference in sensitivity indices}
  \label{tab:sens}
  \begin{tabular}{@{\hskip 0em}lll@{\hskip 0em}}
  \textbf{Benchmark}
  & $\mathbf{x_0\!\rightarrow\!y_1}$        & $\mathbf{x_0\!\rightarrow\!dy_0}$ \\
  \midrule
  Crop-Monitor & 0.00003 & 0.0004 \\
  Car-Straight & 0.0002  & 0.006  \\
  Car-Curved   & 0.00003 & 0.003  \\
  ACAS-Table   & 0.00001 & 0.009  \\
  ACAS-NN      & 0.00002 & 0.061  \\
  \end{tabular}
\end{minipage}
\end{table}

\subsection{Does \NAME{} accurately estimate global sensitivity indices of the
vehicle model?}
\label{sec:eval:sens}

Table~\ref{tab:sens} presents the maximum difference between sensitivity indices
calculated using $M_\mathit{GPC}$ and those calculated using $M'_V$. Column~2 is
for the sensitivity of state variables in time step 1 to those in time step 0
($x_0\!\rightarrow\!y_1$) and Column~3 is for the sensitivity of the change in
the state variables ($x_0\!\rightarrow\!dy_0$ where $dy_0\!=\!y_1\!-\!y_0$). The
sensitivity indices calculated by $M_\mathit{GPC}$ and $M'_V$ match closely.

Figure~\ref{fig:apdx:agbotsens} shows an example visual comparison of
sensitivity indices calculated by \NAME{} and MCS for the sensitivity of the
change in state between time step 0 and 1 to the initial state in time step 0
for Crop-Monitor. The input variables include the heading and distance at time
step 0, and the two components of the raw sample that is transformed into the
perception neural network output distribution sample. The output variables are
the heading and distance at time step 1. There is one sensitivity index
corresponding to each input/output variable pair. The blue solid line shows the
sensitivity calculated analytically by \NAME{}, the blue dotted line shows the
sensitivity calculated empirically using the \NAME{} model, and the red dashed
line shows the sensitivity calculated empirically using $M'_V$. In each subplot,
the X-Axis shows the number of samples used for empirical estimation, while the
Y-Axis shows the calculated sensitivity index. As the number of estimation
samples is increased, the sensitivity indices calculated via estimation converge
towards those calculated analytically. About $10^6$ samples are needed for
convergence.

\begin{figure}
  \centering
  \includegraphics[width=0.9\columnwidth]{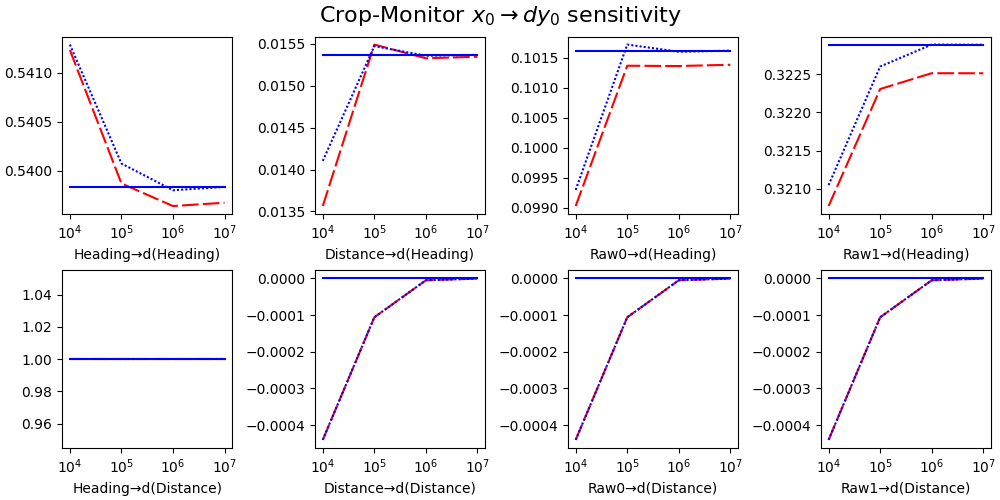}\!
  \caption{Calculated sensitivity index comparison. Blue solid: \NAME{} (analytical method), blue dotted: \NAME{} (empirical method), red dashed: MCS.}
  \label{fig:apdx:agbotsens}
\end{figure}

\textit{In conclusion, \NAME{} enables precise global sensitivity analysis of
vehicle models.}

\subsection{How do the parameters of the \NAME{} perception model affect its
accuracy?}
\label{sec:eval:permodab}

Table~\ref{tab:permodab} describes the effects of changing the perception model
parameters $G$ and $N_i$. Column~1 presents the value of the parameter.
Columns~2-4 present the maximum KS statistic and Wasserstein metric ($\times
100$) for the Crop-Monitor, Car-Straight, and Car-Curved benchmarks,
respectively. Column~5 presents the estimated speedup caused by changing the
parameter value, calculated based on the reduction in the number of images that
must be captured and processed. We exclude the ACAS benchmarks as they do not
use a perception model.

We focus on the cases where the error metrics change by 10\% or more. While the
grid size $G$ can be reduced to $9 \times 9$ without much loss of accuracy,
further reducing it to $7 \times 7$ increases the error for all benchmarks.
Similarly, reducing the number of images captured at each grid point ($N_i$) to
225 does not cause much loss of accuracy, but further reducing it to 100 images
increases the error for all benchmarks. We prefer to err on the side of caution
by collecting 350 images over a $11 \times 11$ grid. The minimal change in
accuracy caused by increasing $G$ from $9 \times 9$ to $11 \times 11$ or
increasing $N_i$ from 225 to 350 also shows that further increases are unlikely
to improve the accuracy of \NAME{}.

\textit{In conclusion, perception model parameters can have a significant
effect on the end-to-end results of \NAME{}, and we must choose them carefully
for optimal accuracy.}

\begin{table}
  \centering
  \small
  \caption{Effect of changing perception model parameters on accuracy. An
  asterisk (*) indicates the primary value used in our evaluation. Changes of
  10\% or more are highlighted.}
  \label{tab:permodab}
  \definecolor{BetterColor}{rgb}{0,.5,0}
  \definecolor{WorseColor}{rgb}{.5,0,0}
  \newcommand{\B}[1]{\textcolor{BetterColor}{\textbf{#1}}}
  \newcommand{\W}[1]{\textcolor{WorseColor}{\textbf{#1}}}
  \begin{tabular}{@{\hskip 0em}lr@{\hskip .2em}c@{\hskip .2em}rr@{\hskip .2em}c@{\hskip .2em}rr@{\hskip .2em}c@{\hskip .2em}rr@{\hskip 0em}}
    \textbf{Parameter} & \multicolumn{9}{c}{$\mathbf{KS_{max}\!\times\!10^2/Wass_{max}\!\times\!10^2}$} & \textbf{Relative} \\
    \textbf{Value} & \multicolumn{3}{c}{\textbf{C-Mon}} & \multicolumn{3}{c}{\textbf{C-Str}}
    & \multicolumn{3}{c}{\textbf{C-Cur}} & \textbf{Speedup} \\
    \midrule
    \multicolumn{11}{l}{Ground truth grid dimensions ($G$)} \\
    \midrule
    $7 \times 7$    & \W{20.3} &/&    2.43  & \W{46.8} &/& \W{9.63} & \W{23.8} &/& \W{3.24} & $2.5\times$ \\
    $9 \times 9$    & \B{12.2} &/&    2.25  &    42.4  &/&    8.39  &    16.2  &/& \B{2.00} & $1.5\times$ \\
    $11 \times 11$* &    13.8  &/&    2.37  &    41.1  &/&    7.97  &    17.0  &/&    2.23  & $1.0\times$ \\
    \midrule
    \multicolumn{11}{l}{Images captured per grid point ($N_i$)} \\
    \midrule
    100             & \W{33.1} &/&    2.34  & \W{49.6} &/& \W{9.77} & \W{18.7} &/& \W{2.59} & $3.5\times$ \\
    225             & \W{16.1} &/&    2.35  &    40.5  &/&    7.91  &    17.0  &/&    2.43  & $1.6\times$ \\
    350*            &    13.8  &/&    2.37  &    41.1  &/&    7.97  &    17.0  &/&    2.23  & $1.0\times$ \\
  \end{tabular}
\end{table}

\subsection{Is \NAME{} \emph{faster} compared to Monte Carlo Simulation?}
\label{sec:eval:speedup}

\begin{table}
  \centering
  \begin{minipage}[t]{0.47\textwidth}
    \centering
    \small
    \caption{Time usage for constructing $M_\mathit{GPC}$}
    \label{tab:eval:constime}
    \begin{tabular}{@{\hskip 0em}lrrr@{\hskip 0em}}
    \textbf{Benchmark} & $\mathbf{t_{dat}}$ & $\mathbf{t_{per/anc}}$ & $\mathbf{t_{GPC}}$ \\
    \midrule
    Crop-Monitor & 8.5h & 1.1s & 1.4s \\
    Car-Straight & 3.3h & 1.1s & 1.4s \\
    Car-Curved   & 3.1h & 1.1s & 1.4s \\
    ACAS-Table   &  N/A & 0.3s & 0.3s \\
    ACAS-NN      &  N/A & 0.3s & 0.4s \\
    \end{tabular}
  \end{minipage}\hfill
  \begin{minipage}[t]{0.48\textwidth}
    \centering
    \small
    \caption{Time usage for state distribution estimation}
    \label{tab:eval:disttime}
    \begin{tabular}{@{\hskip 0em}lrr@{\hskip .2em}l@{\hskip 0em}}
    \textbf{Benchmark} & $\mathbf{t_{MCS}}$ & \multicolumn{2}{l}{$\mathbf{t_{\NAME{}}}$} \\
    \midrule
    Crop-Monitor & 19.5h & 8.5h & ($2.3\times$) \\
    Car-Straight &  6.8h & 3.3h & ($2.1\times$) \\
    Car-Curved   &  6.5h & 3.1h & ($2.1\times$) \\
    ACAS-Table   &  8.0s & 1.1s & ($7.3\times$) \\
    ACAS-NN      & 10.7s & 1.2s & ($8.9\times$) \\
    \end{tabular}
  \end{minipage}
\end{table}

\begin{table}
  \centering
  \small
  \caption{Time usage for sensitivity analysis (excluding $t_\mathit{dat}$ and $t_\mathit{per}$ from Table~\ref{tab:eval:constime})}
  \label{tab:eval:senstime}
  \begin{tabular}{@{\hskip 0em}lrr@{\hskip .2em}lr@{\hskip .2em}l@{\hskip 0em}}
  \textbf{Benchmark} & $\mathbf{t^{emp}_{MCS}}$ & \multicolumn{2}{l}{$\mathbf{t^{emp}_{\NAME{}}}$} & \multicolumn{2}{l}{$\mathbf{t^{ana}_{\NAME{}}}$} \\
  \midrule
  Crop-Monitor & 9.5s & 6.2s & ($1.3\times$) & 11.4s & ($0.7\times$) \\
  Car-Straight & 9.7s & 6.2s & ($1.3\times$) & 11.4s & ($0.8\times$) \\
  Car-Curved   & 9.8s & 6.2s & ($1.3\times$) & 11.3s & ($0.8\times$) \\
  ACAS-Table   & 5.2s & 5.6s & ($0.9\times$) &  3.3s & ($1.6\times$) \\
  ACAS-NN      & 4.8s & 5.3s & ($0.9\times$) &  3.1s & ($1.5\times$) \\
  \end{tabular}
\end{table}

\paragraph{\NAME{} model construction.} Table~\ref{tab:eval:constime} shows the
time required to construct the \NAME{} model. Column~1 shows the benchmark,
Column~2 shows the time required to gather training data for the perception
model as necessary, Column~3 shows the time required to create the perception
and/or ancillary model as necessary, and Column~4 shows the time required to
create $M_\mathit{GPC}$. While creating the perception, ancillary, and GPC
models takes a few seconds, gathering the training data for the perception model
takes several hours. As shown in Section~\ref{sec:eval:permodab}, reducing
perception model parameters such as $|G|$ and $N_i$ can reduce this time, but
can also lead to a reduction in accuracy.

\paragraph{State distribution estimation.} Table~\ref{tab:eval:disttime} shows
the time required by \NAME{} and MCS for state distribution estimation. Column~1
shows the benchmark. Column~2 shows the time required for using $M_V$ for state
distribution estimation. Column~3 shows the total time required by \NAME{} --
this includes the \emph{total} time required to construct $M_\mathit{GPC}$ (from
Table~\ref{tab:eval:constime}) and then evaluate it for state distribution
estimation via Algorithm~\ref{code:app:gpcUseApprox}. Column~3 also shows the
speedup of \NAME{} over MCS. The costly process of gathering and processing
images contributes to over 99\% of $t_\mathit{MCS}$ for the Crop-Monitor and Car
benchmarks. While increasing $N_s$ or $T$ increases $t_\mathit{MCS}$
significantly, the corresponding increase in $t_\mathit{\NAME{}}$ is negligible
as the time required to create the perception model is independent of $N_s$ and
$T$. Consequently, the speedup of \NAME{} for state distribution estimation
increases for longer experiments or higher number of samples. For the ACAS
benchmarks, the control component of $M_V$ contributes to over 90\% of
$t_\mathit{MCS}$. $M_\mathit{GPC}$ is faster than even the dynamics component of
$M_V$, leading to significant speedups for these benchmarks.

\paragraph{Sensitivity analysis.} Table~\ref{tab:eval:senstime} shows the time
required by \NAME{} and MCS for sensitivity analysis. Column~1 shows the
benchmark. Columns~2-3 show the required for calculating \emph{all} sensitivity
indices empirically with $M'_V$ and $M_\mathit{GPC}$ respectively. Column~4
shows the time required calculating all sensitivity indices analytically with
$M_\mathit{GPC}$. Columns~3-4 also show the speedup of \NAME{} for the two
approaches using $M_\mathit{GPC}$. Since both $M'_V$ and $M_\mathit{GPC}$ use
the perception model for the first three benchmarks, we exclude the time
required to train and construct the perception model for those benchmarks. The
replacement of the perception system by the perception model significantly
speeds up $M'_V$ as compared to $M_V$ and also enables vectorization. As this
optimized version is the baseline for sensitivity indices calculation, the
speedup of \NAME{} for this application is lower than that for distribution
estimation. While the analytical and empirical methods for calculating
sensitivity using $M_\mathit{GPC}$ have similar accuracy, different methods are
faster for different benchmarks.

\paragraph{Rapid iteration.} For the first three benchmarks, if the perception
system is altered, \NAME{} must create a new perception and GPC model.
Consequently, while the speedup of \NAME{} stays the same, the total amount of
time saved over MCS increases with each iterative change. If only the vehicle
control or dynamics are altered, then \NAME{} does not need to create a new
perception model. Because gathering data for the perception model is the major
contributor to the runtime of \NAME{} as shown in
Tables~\ref{tab:eval:constime}~and~\ref{tab:eval:disttime}, this allows vehicle
developers to rapidly make changes to the control and dynamics systems of the
vehicle and reanalyze the system with these changes within seconds.

\textit{In conclusion, \NAME{} provides geometric mean speedups of $3.7\times$
for state distribution estimation} and $1.4\times$ for sensitivity analysis.
Increasing the number of samples or time steps or using \NAME{} for rapid
iteration during development further increases this speedup.
\section{Discussion}
\label{sec:disc}

\subsection{Scalability}
\label{sec:disc:scale}

The GPC polynomial model must consider all possible interactions between state
variables, and can suffer from the ``curse of dimensionality'' (combinatorial
explosion of the number of polynomial terms). This is a general limitation of
GPC, for which researchers in engineering applications have proposed solutions
such as low-rank approximation~\cite{konakli2015uncertainty}.

Another solution is to omit higher-order polynomial terms in which multiple
state variables interact~\cite{chaospyTrunc}. Figure~\ref{fig:disc:trunceffects}
shows the effects of this polynomial truncation approach for an order 4 GPC
model. For both plots, the X-Axis shows the number of dimensions in the state
space. The Y-Axis of the left plot shows the number of polynomial terms in the
constructed GPC model, and the Y-Axis of the right plot shows the model
construction time. The dashed red lines and the solid blue lines show the
results for the full polynomial and the truncated polynomial, respectively. The
speedup from truncation increases rapidly with dimensionality. Critically, this
method retains lower-order interaction terms, ensuring that the constructed GPC
model can still account for interactions between state variables (albeit to a
lesser degree of fidelity).

\begin{figure}
\centering
\includegraphics[width=0.5\columnwidth]{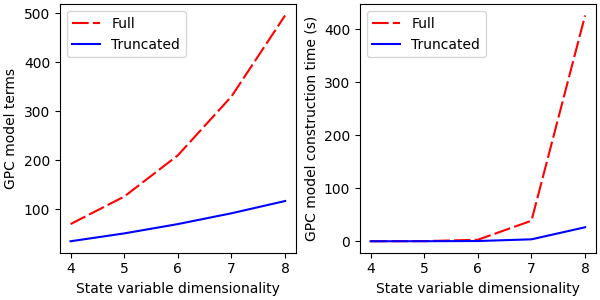}
\caption{Effect of truncation on GPC model terms and construction time for an order 4 GPC model.}
\label{fig:disc:trunceffects}
\end{figure}

\subsection{Limitations}
\label{sec:disc:limit}

\NAME{} samples the provided environment distribution $\mathcal{D}_E$ when
generating training data for the perception model in
Algorithm~\ref{code:app:perTrain}. If \NAME{} uses the perception model created
for $\mathcal{D}_E$ to create a GPC model for a scenario with a significantly
different environment distribution $\mathcal{D}_E'$, it can lead to a reduction
in accuracy of the final GPC model. To prevent this issue, we must carefully
choose $\mathcal{D}_E$ to represent the distribution of environments the vehicle
is expected to operate in. Even if creating a new perception model is necessary,
we may be able to partially reuse existing training data from $\mathcal{D}_E$
that also fits into $\mathcal{D}_E'$, reducing the additional time required to
create the new perception model.

To construct the GPC model in Algorithm~\ref{code:app:gpcCons}, all component
distributions of the multivariate distribution $J$ must have corresponding
orthogonal polynomials. This is true for a wide variety of distributions such as
the uniform, normal, beta, and gamma distributions. We must ensure that all
input variables for the GPC model are independent and are instances of these
distributions. The wide variety of supported distributions is usually sufficient
for modeling relevant input variables.

As GPC produces a polynomial model, its accuracy is limited when modeling
functions with discontinuities or limited
differentiability~\cite[Theorem~3.6]{XiuGPCBook}. Practically, this means that
the GPC model will have a systemic bias that can be reduced, but not eliminated.
Our evaluation shows that \NAME{}'s hyperparameters must be carefully chosen to
maximize accuracy while still providing speedups over MCS. However, once the
ideal hyperparameters are found for modeling a particular vehicle system, it is
possible to reuse them when making iterative changes to the vehicle model.

We demonstrate the applicability of \NAME{} to autonomous vehicle scenarios with
various structures and associated challenges. However, the applicability and
efficiency of \NAME{} for arbitrary autonomous vehicle scenarios is an open
question and an interesting topic for future investigation.

\subsection{Alternative formulations}
\label{sec:disc:alts}

The GPC model construction process in Algorithm~\ref{code:app:gpcCons} evaluates
the abstracted vehicle model at specific quadrature nodes. Instead of
constructing a perception model as in Section~\ref{app:permod}, we attempted to
directly calculate and use the actual perception neural network output
distributions at these quadrature nodes. However, if the quadrature nodes change
(e.g., by changing the state distribution or order of GPC), then new images must
be captured and processed for the new quadrature nodes. In contrast, the
perception model can be directly reused as it is agnostic to the quadrature
nodes.

We also experimented with replacing only parts of $M'_V$ with GPC as follows:
first, we replaced the perception system with the perception model to create
$M'_V$ as in Section~\ref{app:permod}. Then, instead of replacing all of $M'_V$
with a GPC model as in Section~\ref{sec:app:gpccons}, we replaced only the
vehicle control and dynamics systems (as the perception system had already been
replaced by a polynomial model $M_\mathit{per}$). While the accuracy of this
approach was the same as our main approach, the partially replaced model was
about $3\times$ slower than the fully replaced model during evaluation for all
benchmarks.
\section{Related Work}
\label{sec:relwork}

\paragraph{Analysis and verification of vehicle systems.}
DryVR~\cite{DryVR} is a system for verifying the safety of vehicle models that
consist of a whitebox mode-switching control system composed with a blackbox
dynamics system. DryVR obtains multiple samples from the blackbox dynamics and
uses it give an over-approximate guarantee on the reachable set of states. In
comparison, \NAME{} focuses on perception and control systems which include
blackbox components such as neural networks.

Jha et al.~\cite{autonomyCCTL} and DeepDECS~\cite{calinescu2022discreteevent}
analyze perception system uncertainty to create a correct by synthesis
controller that satisfies a temporal logic safety constraint. \NAME{} allows us
to estimate the safety of existing, well known controllers for broad
applications. Althoff et al.~\cite{6784493} focus on online verification of
vehicle safety properties to adapt to unique traffic situations, assuming an
upper bound on the noise of the perception system. In contrast, \NAME{} allows
the programmer to directly specify the perception system, and handles the
dynamic nature of perception error through the perception model.

Musau et al.~\cite{9763623} complement complex neural network controllers with a
safety controller which takes over if a runtime reachability analysis detects a
potential collision. \NAME{}'s state distribution estimates can be used to
determine the degree of reliability of the neural network controller and
therefore the optimal criteria for switching to the safety controller.

Yang et al.~\cite{9797620} propose a runtime system for detecting environmental
conditions that were not part of the training data for a pre-trained perception
system. Cheng et al.~\cite{10.1007/978-3-030-01090-4_8} ensure that the training
data covers different combinations of environmental factors, based on their
relative importance. Scenic~\cite{Scenic} is a probabilistic language for
specifying scenes in a virtual world for generating training images for
perception neural networks in various environments. These works can be used when
sampling images and/or prioritizing environments for training \NAME{}'s
perception model in order to ensure that \NAME{}'s results are valid for the
entire range of expected environments. For our evaluation, we used simple but
realistic distributions of environmental parameters when generating images for
training the perception model.

\paragraph{System simulation and modeling.} Kewlani et al.~\cite{GPCforDynamics}
create and use GPC surrogate models of vehicle \emph{dynamics} components. Lin
et al.~\cite{mlsurrogate} do the same using small neural networks as surrogate
models. Replacing only vehicle dynamics with GPC surrogates leads to negligible
speedup as the perception and control components of the vehicle model
contributed to 90\% or more of the runtime of the MCS analysis of our
benchmarks. \NAME{} creates GPC surrogate models of \emph{complete} vehicle
systems -- the GPC model also replaces the expensive perception and control
components. ARIsTEO~\cite{menghi2020approximation} uses abstraction refinement
to create surrogate models of cyber-physical systems with low dimensional
inputs. \NAME{} handles high-dimensional image inputs by first creating a
perception model, and then creates a surrogate model of the
reduced-dimensionality abstract vehicle model.

Li et al.~\cite{arxiv.2211.12733} create small neural network surrogate models
to calculate a \emph{fitness function} for an autonomous driving system in
various traffic scenarios. \NAME{}'s surrogate models instead estimate the state
distribution of the vehicle over time -- information which can be used to
calculate various fitness metrics. Michelmore et al.~\cite{uqine2eadc} evaluate
the safety of end-to-end Bayesian Neural Network (BNN) controllers. While
\NAME's perception model focuses on replacing regression DNNs, vehicle models
with BNNs are an interesting target for extending \NAME{}.

\paragraph{Simplification of complex perception and control systems.} Cheng et
al.~\cite{regressionAltNN} explore the correspondence between neural networks
and polynomial regressions. Unlike their approach, \NAME{}'s perception model
only needs to predict the \emph{distribution} of neural network outputs, instead
of individual input-output relationships.

Several approaches focus on sound and complete verification of safety properties
of ACAS-Xu neural networks (e.g., that the network will not give a COC advisory
when an intruder aircraft is directly ahead on a collision course~\cite{NNV,
NNEnum}). Others focus on providing probabilistic guarantees (e.g., bounding the
probability that similar inputs will result in different
advisories~\cite{ConverseFGP20}, or the probability of violating the safety
property described above~\cite{tranquantverif}). Unlike \NAME{} and the
approaches described below, these approaches focus on checking safety properties
of the networks \emph{in isolation}, that is, they do not take into account how
the aircraft moves over time.

Several other recent approaches (e.g.,~\cite{horizontalCAS, BakNFM22,
LopezJAT22}) focus on deterministically verifying properties of a
\emph{closed-loop system} containing ACAS-Xu neural networks and aircraft
dynamics. In particular, Bak and Tran~\cite{BakNFM22} checked the safety of the
neural networks from~\cite{juliandnncomp} in 32 minutes on a 128 core machine.
In contrast to those approaches, \NAME{} does simulation-based probabilistic
testing of ACAS-Xu neural networks coupled with aircraft dynamics. While \NAME{}
can find unsafe scenarios, it cannot provide non-probabilistic safety
guarantees. Julian et al.~\cite{juliandnncomp} performed 1.5 million simulations
to test their ACAS-Xu neural network, and Owen and Kochenderfer~\cite{7777959}
performed 10 million simulations to determine whether horizontal or vertical
advisories are safer in different states. We view \NAME{} as a useful tool for
speeding up such simulations and sanity checks of the closed-loop aircraft
system when making experimental changes to it, before moving on to full
verification.

Ghosh et al.~\cite{arxiv.1911.01523} iteratively synthesize perception models
and controllers guided by counterexamples to temporal logic safety properties.
Hsieh et al.~\cite{chiaoPaper} create a perception model where the mean is
calculated using piecewise linear regression and the allowable variance is
calculated based on the controller code using program analysis tools like CBMC.
Unlike these works, \NAME{}'s perception model is independent of the controller,
and can be reused without requiring additional sampling when iterative changes
are made to the controller during development.

\section{Conclusion}
\label{sec:conc}

We present \NAME{}, the first approach for using Generalized Polynomial Chaos
(GPC) to analyze a complete autonomous vehicle system which composes complex
perception and/or control components with vehicle dynamics. \NAME{} first
creates a model of the perception subsystem and uses it to generate a surrogate
model of the entire vehicle system. \NAME{} accurately models the system while
being $3.7\times$ faster on average for state distribution estimation and
$1.4\times$ faster on average for sensitivity analysis than Monte Carlo
Simulation on five scenarios used in agricultural vehicles, self driving cars,
and unmanned aircraft. We anticipate that \NAME{} can augment or partially
replace simulation-based workflow in tasks that require iterative modifications
or continuous re-execution of the model -- for example, during vehicle
development.

\bibliographystyle{ACM-Reference-Format}
\bibliography{references}

\end{document}